\documentclass[11pt]{article}
\usepackage[hmargin=1in,vmargin=1in]{geometry}
\usepackage{hchang}

\usepackage{multirow}
\usepackage{xspace}
\usepackage{makecell}
\usepackage{comment}
\usepackage[hyperref]{ntheorem}
\usepackage{thm-restate}
\usepackage{cleveref}

\nolinenumbers

\theoremseparator{.}
\theorembodyfont{\slshape}
\newtheorem{theorem}{Theorem}[section]
\newtheorem{lemma}[theorem]{Lemma}

\newtheorem{observation}[theorem]{Observation}

\newtheorem{question}[theorem]{Question}
\newtheorem{conjecture}[theorem]{Conjecture}
\crefname{conjecture}{conjecture}{conjectures}
\crefname{observation}{observation}{observations}

\theorembodyfont{\upshape}
\newtheorem{remark}[theorem]{Remark}
\newtheorem{definition}[theorem]{Definition}

\def\eps{\varepsilon}

\def\reals{\mathbb{R}}

\newcommand{\CMP}{\lhd} 
\newcommand{\llhd}{\mathrel{\lhd\mkern-8mu\lhd}}
\newcommand{\lllhd}{\mathrel{{\lhd}\mkern-9mu{\lhd}\mkern-9mu{\lhd}}}

\newcommand{\OO}{\tilde{O}}

\newcommand{\cube}[1]{\llbracket #1\rrbracket}
\newcommand{\ball}[1]{\llparenthesis #1 \rrparenthesis}

\newcommand{\shadow}{\textrm{shadow}}
\newcommand{\Diameter}{\textsc{Diameter}\xspace}

\newcolumntype{L}{>{$}l<{$}}

\renewcommand{\arraystretch}{1.3}

\newcounter{ctr}
\loop
 \stepcounter{ctr}
 \expandafter\edef\csname c\Alph{ctr}\endcsname{\noexpand\mathcal{\Alph{ctr}}}
\ifnum\thectr<26
\repeat

\title{Charting the Diameter Computation Landscape of Geometric Intersection Graphs in Three Dimensions and Higher}
\author{Timothy M. Chan%
    \thanks{Siebel School of Computing and Data Science,
    University of Illinois at Urbana-Champaign. Email: tmc@illinois.edu.  Supported by NSF grant CCF-2224271.}
    \and
    Hsien-Chih Chang%
    \thanks{Department of Computer Science, Dartmouth College. Email: hsien-chih.chang@dartmouth.edu.}
    \and
    Jie Gao%
    \thanks{Department of Computer Science, Rutgers University. Email: jg1555@rutgers.edu. Gao would like to acknowledge NSF support through CNS-2515159, IIS-2229876, DMS-2220271,  DMS-2311064, CCF-2208663, CCF-2118953.}
    \and
    S\'andor Kisfaludi-Bak%
    \thanks{Department of Computer Science, Aalto University, Finland. Email: 
    sandor.kisfaludi-bak@aalto.fi. Supported by the Research Council of
    Finland, Grant 363444.}
    \and
    Hung Le%
    \thanks{Manning CICS, UMass Amherst. Email: hungle@cs.umass.edu. Supported by NSF grants CCF-2517033 and CCF-2121952, NSF CAREER Award CCF-2237288, and a Google Faculty Research Award.}  
    \and 
    Da Wei Zheng%
    \thanks{Institute of Science and Technology Austria, Klosterneuburg, Austria. Work in this paper was started while at the Siebel School of Computing and Data Science, University of Illinois at Urbana-Champaign.}
}

\begin{document}
\maketitle
\begin{abstract}
Recent research on computing the diameter of geometric intersection graphs has made significant strides, primarily focusing on the 2D case~\cite{duraj2023better,CGL24,ChanCGKLZ25} where truly subquadratic-time algorithms were given for simple objects such as unit-disks and (axis-aligned) squares. 
However, in three or higher dimensions, there is no known truly subquadratic-time algorithm for any \emph{intersection graph of} non-trivial objects, even basic ones such as unit balls or (axis-aligned) unit cubes.  
This was partially explained by the pioneering work of Bringmann et al.~\cite{Bringmann2022-me} which gave several truly subquadratic lower bounds, notably for unit balls or unit cubes in 3D when the graph diameter $\Delta$ is at least $\Omega(\log n)$, hinting at a pessimistic outlook for the complexity of the diameter problem in higher dimensions. 
In this paper, we substantially extend the landscape of diameter computation for objects in three and higher dimensions, giving a few positive results. Our highlighted findings include:

\begin{enumerate}
    \item  A truly subquadratic-time algorithm for deciding if the diameter of unit cubes in 3D is at most 3 (\Diameter-3 hereafter), the first algorithm of its kind for objects in 3D or higher dimensions.  Our algorithm is based on a novel connection to pseudolines, which is of independent interest.
    \item  A truly subquadratic time lower bound for \Diameter-3 of unit balls in 3D under the Orthogonal Vector (OV) hypothesis, giving the first separation between unit balls and unit cubes in the small diameter regime.  Previously, computing the diameter for both objects was known to be truly subquadratic hard when the diameter is $\Omega(\log n)$~\cite{Bringmann2022-me}.
    \item A near-linear-time algorithm for \Diameter-2 of unit cubes in 3D, generalizing the previous result for unit squares in 2D~\cite{Bringmann2022-me}. 
    \item A truly subquadratic-time algorithm and lower bound for \Diameter-2 and \Diameter-3 of rectangular boxes (of arbitrary dimension and sizes), respectively.  
\end{enumerate}
\end{abstract}
\thispagestyle{empty}

\newpage
\tableofcontents
\thispagestyle{empty}

\newpage
\setcounter{page}{1}
\section{Introduction}

Computing the diameter of sparse graphs is \EMPH{quadratic hard}: under the Strong Exponential Time Hypothesis (SETH), there is no $O(n^{2-\varepsilon})$ algorithm for 
distinguishing diameter 2 vs.\ 3 of graphs with $n$ vertices and $\OO(n)$\footnote{$\OO(\cdot)$ hides polylogarithmic factors in $n$.} edges~\cite{Roditty2013-zr}. That is, the trivial algorithm for computing the diameter by doing BFS from every vertex is essentially optimal. Since then, the major research focus has been on substantially beating the BFS-based algorithm for structural families of graphs, especially planar/minor-free graphs and geometric intersection graphs.  For planar/minor-free graphs, truly subquadratic algorithms were known~\cite{Cabello2018-gz,Gawrychowski2018-zy, Ducoffe2019-tp,le2023vc,ChanCGKLZ25}.  For geometric intersection graphs, the complexity of computing the diameter remains poorly understood, due to the sheer diversity of geometric objects underlying the graphs and the fact that geometric intersection graphs can have $\Omega(n^2)$ edges.   

The pioneering work of Bringmann \etal~\cite{Bringmann2022-me} studied the diameter of intersection graphs of several types of objects.  In dimension three and above, they considered (axis-parallel) hypercubes\footnote{In this paper, hypercubes and boxes are axis-parallel by default, unless noted otherwise.} and balls. 
Let \EMPH{\Diameter-$\Delta$} be the problem to decide if a given input graph has diameter at most $\Delta$ or at least $\Delta+1$.
They showed that it is quadratic hard to solve: (1) \Diameter-$2$ for hypercubes in $\reals^{12}$ under the (3-uniform 6-)Hyperclique Hypothesis; 
(2) \Diameter-$3$ for (not necessarily axis-aligned) unit segments and equilateral triangles in $\reals^2$ under the Orthogonal Vector (OV) Hypothesis;  and 
(3) \Diameter-$\Omega(\log n)$ for unit balls and unit cubes in $\reals^3$, and axis-parallel line segments in $\reals^2$, all
under the OV Hypothesis. 

The hardness results can be interpreted as indicating that the radius-$\Delta$ neighborhood balls of each geometric shape class are sufficiently complex and expressive, capable of encoding the hard instances used in various fine-grained reductions.
One possible way to quantify the complexity of the set system of neighborhood balls is the notion of \emph{VC-dimension}.
Given a set system  $(U, \mathcal{F})$ with a ground set $U$ and a family $\mathcal{F}$ of subsets of $U$, its \EMPH{VC-dimension} is the cardinality of the largest $S\subseteq U$ such that $S$ is \emph{shattered} by $\mathcal{F}$---%
for every $S'\subseteq S$, there is some $X\in \mathcal{F}$ such that $X\cap S = S'$. 
We say that a graph $G$ has \EMPH{distance VC-dimension} at most $d$ if the set system of radius-$r$ neighborhood balls $(V_G, \{N^r[v]\}_{r \geq 0})$ has VC dimension at most $d$.
Planar graphs (more generally, minor-free graphs)~\cite{Chepoi2007,Bousquet2015,Li2019-li,ducoffe2022diameter,le2023vc} and 
intersection graphs of pseudo-disks~\cite{AbuAffash2021,duraj2023better,CGL24} both have bounded distance VC-dimension,
whereas intersection graphs of unit segments and equilateral triangles do not.
Recently, Chan \etal~\cite{ChanCGKLZ25} gave truly-subquadratic-time diameter algorithms for squares and unit-disks in the plane---both shapes being instances of pseudo-disks, thus with finite distance VC-dimension. Together, these results hinted at \emph{finite VC-dimension} as an overarching property for fast diameter algorithms.

We undertake a comprehensive study of computing diameter on geometric intersection graphs. 
Our results are presented in a pair of papers. 
While the companion paper~\cite{Anon2D25} focuses solely on the 2D case, this paper is devoted to three or higher dimensions.  In the higher-dimensional case, there is no known truly subquadratic time algorithm for any non-trivial type of objects, even for basic ones such as unit balls or unit cubes. Upper bound techniques for the 2D cases~\cite{Bringmann2022-me,duraj2023better,CGL24,ChanCGKLZ25} heavily rely on planarity in various places, most notably in bounding the VC dimension using the non-planarity of $K_5$~\cite{duraj2023better,CGL24,ChanCGKLZ25}. These planarity-specific techniques, in addition to the negative results by Bringmann \etal~\cite{Bringmann2022-me}, hint at a pessimistic outlook for the complexity of diameter problems in higher dimensions. 

\begin{question}
\label{ques:nontrivial} 
Does a truly subquadratic time algorithm for computing the diameter exist for any natural class of geometric intersection graphs in 3D or higher dimensions?
\end{question}

As a starting point, the hardness landscape for unit cubes seems peculiar.
Recall that Bringmann \etal~\cite{Bringmann2022-me} showed that in 12D, distinguishing diameter 2 versus 3 is already quadratic hard (under the Hyperclique Hypothesis).  
However in 3D, the quadratic hardness does not kick in until $\Delta \ge \Omega(\log n)$.
Duraj, Konieczny, and Pot{\k e}pa~\cite{duraj2023better} solves the 2D case in time $\OO(\Delta n^{7/4})$ for \Diameter-$\Delta$, which was later extended by \cite{ChanCGKLZ25} to solve \Diameter in $\OO(n^{2-1/8})$ time for general $\Delta$.
What about unit cubes in 3D?
The same question can be asked about unit balls, which typically behave similarly to unit cubes. 

\begin{question}
\label{ques:balls-n-cubes} 
What is the complexity landscape of \Diameter-$\Delta$ for unit cubes and unit balls in constant dimensions?
Do radius-$\Delta$ neighborhood balls of their intersection graphs have bounded VC-dimension, in particular in 3D?
\end{question}

\begin{table}[!t]
\footnotesize\sffamily\sansmath
\renewcommand{\arraystretch}{1.3}
\centering
\begin{NiceTabular}{c c c c}
\multicolumn{2}{c}{
\textbf{Graph class}}  & \textbf{Lower bound} & \textbf{Upper bound}  \\
\hline
\hline
\multirow{2}{*}{Unit balls}
  & 2D &  & 
  \shortstack{
  $O^*(n^{2-1/18})$ for general $\Delta$ \cite{ChanCGKLZ25}  \\    
  $O^*(n^{2-1/9})$ for $\Delta=O(1)$ \cite{ChanCGKLZ25}   \\
  $\OO(n^{4/3})$ for $\Delta=2$ (2D Paper~\cite{Anon2D25})
  }
  \\ 
  \cline{2-4}
  & & $\Omega^*(n^2)$ for general $\Delta$ {\tiny (OV)} \cite{Bringmann2022-me} &
  \\ 
  & 3D & \cellcolor{Highlight} $\Omega^*(n^2)$ for $\Delta=3$ {\tiny (OV)} (Thm.~\ref{thm:3D-ball}) &
  \\   
  \cline{2-4}
  & $d$D & \cellcolor{Highlight} $\Omega^*(n^2)$ for $\Delta=2$ \& $d=7$ {\tiny (combK4)} (Thm.~\ref{thm:7D-ball}) &
  \\ 
\hline
\hline
\multirow{4}{*}{\shortstack{Unit\\ hypercubes}} 
  & 2D &  & \shortstack{$O(n\log n)$ for $\Delta=2$ \cite{Bringmann2022-me}\\
       $O^*(n^{7/4})$ for $\Delta=O^*(1)$ \cite{duraj2023better}\\
       $\OO(n^{2-1/8})$ for general $\Delta$ \cite{ChanCGKLZ25}\\
        $O^*(n)$ for any constant $\Delta$ (2D Paper~\cite{Anon2D25})}\\
  \cline{2-4}
  & & & \cellcolor{Highlight} $\OO(n)$ for $\Delta=2$ (Thm.~\ref{thm:3Dunitcube:diam2})
  \\
  & 3D & $\Omega^*(n^2)$ for general $\Delta$ {\tiny (OV)} \cite{Bringmann2022-me} &  
        \cellcolor{Highlight} $\OO(n^{2-1/13})$ for $\Delta=3$ (Thm.~\ref{thm:unitcube3d:diam3})
  \\ 
  \cline{2-4}
  & 4D & \cellcolor{Highlight} $\Omega^*(n^2)$ for $\Delta=3$ {\tiny (OV)} (Thm.~\ref{thm:4D-cube}) &  
  \\ 
  \cline{2-4}
  & & $\Omega^*(n^2)$ for $\Delta=2$ \& $d=12$ {\tiny (3H6)} \cite{Bringmann2022-me}
  \\ 
  & & \cellcolor{Highlight} $\Omega^*(n^2)$ for $\Delta=2$ \& $d=10$ {\tiny (3H6)} (Thm.~\ref{thm:10D-cube})
  \\ 
  & $d$D & \cellcolor{Highlight} $\Omega^*(n^2)$ for $\Delta=2$ \& $d=6$ {\tiny (combK4)} (Thm.~\ref{thm:6D-cube}) &
  \\ 
\hline
\hline
\multirow{3}{*}{Cubes}
  & 2D & & \shortstack{$O^*(n^{2-1/12})$ for general $\Delta$ \cite{ChanCGKLZ25}}\\
  \cline{2-4}
  & & $\Omega^*(n^2)$ for general $\Delta$ {\tiny (OV)} \cite{Bringmann2022-me}
  \\
  & 3D & \cellcolor{Highlight} $\Omega^*(n^2)$ for $\Delta=3$ {\tiny (OV)} (Thm.~\ref{thm:3D-cube}) & \cellcolor{Highlight} $\OO(n^{2-1/5})$ for $\Delta=2$ (Thm.~\ref{thm:cube3d:diam2})
  \\
\hline  
\hline
\multirow{2}{*}{Boxes} 
  & 2D & \cellcolor{Highlight} $\Omega^*(n^2)$ for $\Delta=3$ {\tiny (OV)} (Thm~\ref{thm:2D-rect:diam3}) &
  \cellcolor{Highlight} $\OO(n^{7/4})$ for $\Delta=2$ (Thm~\ref{thm:rect:diam2})
  \\
  \cline{2-4}
  & 3D & & \cellcolor{Highlight} $\OO(n^{2-1/6})$ for $\Delta=2$ (Thm.~\ref{thm:box3d:diam2})
  \\
\hline
\hline
\end{NiceTabular}
\caption{
Previous and new time bounds for deciding whether an intersection graph of geometric objects has diameter at most $\Delta$.  All hypercubes/boxes are axis-aligned. New results are shown in bold.
Conditional lower bounds marked ``(OV)'', ``(3H6)'', and ``(combK4)'' assume the Orthogonal Vectors hypothesis,
the $3$-uniform $6$-hyperclique hypothesis, and the combinatorial 4-clique hypothesis, respectively, where the latter is for combinatorial algorithms (all upper bounds are obtained from combinatorial algorithms). See \Cref{appendix:hypothesis} for definitions.  ``2D Paper'' refers to the companion paper in the pair~\cite{Anon2D25}.}
\label{table:results-3D}
\end{table}

\subsection{Our Contributions}
\label{subec:contribution}

We provide a comprehensive study on basic objects in higher dimensions;
our results are summarized in \Cref{table:results-3D}. 
As corollaries, we make significant progress towards Question~\ref{ques:nontrivial} and Question~\ref{ques:balls-n-cubes}. 
Below, we discuss each of the contributions in more detail.

\paragraph{Unit cubes versus unit balls.}
Our first result is a separation between unit (hyper)cubes and unit balls: \Diameter-3 for unit balls is hard, whereas \Diameter-3 for unit cubes admits a truly subquadratic-time algorithm.  

\begin{theorem}
\label{thm:3Dcube-vs-ball} 
There is no $O(n^{2-\eps})$ algorithm for  \Diameter-3 for \ul{unit balls} in $\reals^3$ under the OV hypothesis for any $\eps > 0$, while for \ul{unit cubes} in $\reals^d$ for $d \ge 3$, \Diameter-3 and \Diameter-2  can be solved in $\OO(n^{2-1/13})$ and $\OO(n)$ time, respectively. 
\end{theorem}

Our lower bound for  \Diameter-3  of unit balls is by reducing from \Diameter-2 of sparse tripartite graphs $G = (A\cup B\cup C, E)$, which is quadratic hard under the OV hypothesis.  
Our first idea for the reduction would be to create a unit ball for each vertex in $A\cup B\cup C$, as well as a unit ball for each edge in $E$. 
We encode edges in the original graph $G$ by intersection: a vertex $u \in A\cup B\cup C$ is incident to an edge $uv\in E$ if and only if the unit balls corresponding to $u$ and $uv$ intersect. 
Then a length-2 path $(a\rightarrow b\rightarrow c)$ in $G$ will correspond to a length-4 path in the intersection graph $ (\ball{a}\rightarrow \ball{ab}\rightarrow \ball{b}\rightarrow \ball{bc} \rightarrow \ball{c})$, where $\ball{x}$ is the unit ball corresponding to (a vertex or an edge) $x$. With this idea, we could obtain hardness of \Diameter-4. 
To obtain a lower bound for \Diameter-3, we do not create unit balls for the middle set of vertices $B$; the path would then be  $(\ball{a}\rightarrow \ball{ab}\rightarrow \ball{bc}\rightarrow \ball{c})$. 
The difficulty is to guarantee that $\ball{ab}$ only intersects $\ball{bc}$, not $\ball{\Tilde{b}c}$ for some other vertex $\Tilde{b} \in B$; in the \Diameter-4 case, this is enforced by the ball $\ball{b}$. Here, we achieve the  guarantee by introducing an angle 
parameter $\delta$ 
to control the intersections of the balls and prevent unwanted edges from being added. See \Cref{sec:diameter3-balls}.

The algorithms for \Diameter-2 and 3 for unit cubes in \Cref{thm:3Dcube-vs-ball} are perhaps most interesting. 
First, we show that \emph{the VC-dimension of the 2- and 3-neighborhoods of 3D unit-cubes is bounded by a constant}. 
Here we depart from existing techniques for the 2D case~\cite{duraj2023better,CGL24,ChanCGKLZ25}, which rely on the non-planarity of $K_5$. 
In a small region, unit cubes behave like orthants, and so the intersection between unit cubes corresponds to a \EMPH{dominance} relation.  We divide the sequence of dominance relations on the length-2 and length-3 path into a number of special cases, bound the VC-dimension of each by 1 or 2 via direct arguments, and then bound the overall VC-dimension by combining subsystems together.
See \Cref{sec:unitcube3d:diam2:VC} and
\Cref{sec:unitcube3d:diam3:VC} for details.

Our VC-dimension bound, in combination with existing techniques~\cite{ChanCGKLZ25}, is enough to imply truly subquadratic time algorithms for \Diameter-2 and \Diameter-3, but because of large constants, the running time would be awfully close to quadratic (about $O(n^{1.999998})$ for \Diameter-3).
We further build on the proof ideas to give ultra efficient algorithms:
In case of \Diameter-2, because the subsystems from various cases actually all have VC-dimension~1,
we are able to use orthogonal range searching, in combination with divide-and-conquer, to
achieve \emph{near-linear time}.  
(Large constants appear in the logarithmic factors, instead of in the main polynomial factor.)
For \Diameter-3, because the subsystems from various cases actually turn out to be special kinds of 2-dimensional
systems corresponding to \emph{pseudoline arrangements}, we are able to use known range searching techniques, in combination with a grid approach, to achieve a more reasonable running time
$\OO(n^{2-1/13})$.  
(We find it surprising that pseudolines turn out to be relevant for a problem about 3D unit cubes!)
The technical details are intricate, as we need to work with abstract pseudolines that are only implicitly represented via certain oracles.
See \Cref{sec:unitcube3d:diam2} and
\Cref{sec:unitcube3d:diam3} for details.

Our constant bound on the VC-dimension of the 2- and 3-neighborhoods of 3D unit cube graphs suggests an intriguing open question of independent interest to combinatorial geometers: do $r$-neighborhoods similarly have bounded VC-dimension for all constant $r>3$?

\begin{conjecture}[VC-Dimension of Unit-Cube Graphs]
\label{conject:VC-cubes} 
There exists a function $f$ such that the VC dimension of $(V, \{N^{r}[v]\}_{v\in V, r\geq 0})$ is at most $f(r)$, where $N^{r}[v]$ is the $r$-neighborhood of~$v$ in the unit cube graph $G= (V,E)$. 
\end{conjecture}

Combining with~\cite{ChanCGKLZ25}, Conjecture \ref{conject:VC-cubes} would imply that unit-cube graphs of constant diameter admit a truly subquadratic time algorithm for computing diameter, which 
is a natural class of intersection graphs motivated by Question~\ref{ques:nontrivial}. 

\paragraph{Unit-hypercubes in 4$^{\text{+}}$D.}
Next, we show that \Diameter-3 is quadratic hard for unit hypercubes in 4D. Therefore, there is no longer a meaningful separation between unit hypercubes and unit balls in dimension 4. When the dimension is even higher, say at least 6D, we can 
get hardness results for \Diameter-2, under the \EMPH{$3$-uniform $6$-hyperclique hypothesis (3H6)} (\Cref{def:3H6}) or \EMPH{combinatorial 4-clique (combK4)} hypothesis (\Cref{def:combK4}),
by refining Bringmann \etal's previous proof~\cite{Bringmann2022-me} which required 12 dimensions.

\begin{theorem}\label{thm:cube-vs-ball} Let $\eps \in (0,1)$ be any given parameter. There is no $O(n^{2-\eps})$ algorithm for:
\begin{enumerate}
    \item  \Diameter-3 for  \EMPH{unit hypercubes in $\reals^4$}  under the OV hypothesis.
     \item  \Diameter-2 of   \EMPH{unit balls in $\reals^7$} under the CombK4 hypothesis.  
      \item  \Diameter-2 for \EMPH{unit hypercubes in $\reals^6$ and $\reals^{10}$} under CombK4 / 3H6 hypotheses, resp.
    \end{enumerate}
\end{theorem}

\paragraph{Cubes and Boxes.}
In 2D, the intersection graphs of (non-unit) squares admit a truly subquadratic-time algorithm for computing the diameter (of any value)~\cite{ChanCGKLZ25}. In 3D, we show that for (non-unit) cubes, \Diameter-3 is quadratic-hard; this, along with the algorithms in \Cref{thm:3Dcube-vs-ball}, gives a separation between cubes and their unit counterparts.  
In contrast, we show that \Diameter-2 of (non-unit) cubes is solvable in truly subquadratic time. Even for the more general case of 3D boxes, we also obtain a truly subquadratic time algorithm. 
(Note that this result is new even for 2D rectangles!) 
Instead of relying on VC-dimension, this subquadratic algorithm uses a different, grid-based approach.
Grid-based approaches have been used before for solving other problems about boxes~\cite{Chan23,ChanHY23}, but one  challenge arises from the fact that 3D boxes (or cubes) have quadratic \emph{union complexity}, unlike 3D unit cubes or orthants.  We show that in the case of diameter 2, some of the sides are irrelevant and some boxes can actually be replaced by orthants.  See \Cref{sec:box3d:diam2:alg}.

\begin{theorem}\label{thm:boxes} There is no $O(n^{2-\eps})$ algorithm for \Diameter-3 for (non-unit) cubes in $\mathbb{R}^3$ under the OV hypothesis for any $\eps > 0$. On the other hand, \Diameter-2 can be solved in time $\OO(n^{2-1/5})$ for cubes and $\OO(n^{2-1/6})$ for boxes in $\mathbb{R}^3$. 
\end{theorem}

\paragraph{Finite VC-dimension as a sufficient condition?}
Our collection of results for many types of objects in 2D (in the companion paper~\cite{Anon2D25}) and higher-dimensional cases reinforces the same pattern: truly subquadratic-time algorithms are mostly for graphs with bounded VC dimension. We conjecture that these are special cases of a broader phenomenon.

\begin{conjecture}[VC-Dimension vs Diameter Conjecture]
\label{conject:diameter} 
\Diameter-$\Delta$ of the intersection graph of low-complexity geometric objects (regardless of the dimension) can be solved in truly subquadratic time if the VC dimension of $\Delta$-neighborhoods is bounded by a constant. 
\end{conjecture}
Here, by low-complexity, we mean that the geometric object has constant description complexity.
A positive resolution of Conjecture~\ref{conject:diameter} would provide a powerful algorithmic tool for solving the diameter problem on geometric intersection graphs. 
A non-trivial test case for Conjecture~\ref{conject:diameter} is (non-unit) disks in 2D: the distance VC dimension is at most 4~\cite{CGL24} whereas currently there is no known truly subquadratic algorithm even for \Diameter-2. 

It is tempting to strengthen Conjecture~\ref{conject:diameter} to an ``if and only if'' statement; however, our truly subquadratic algorithm for \Diameter-2 of boxes in \Cref{thm:boxes} provides a counter-example:
the intersection graphs of diameter 2 for boxes and rectangles have unbounded VC-dimension. 
(For proof of unbounded VC-dimension, see \Cref{rmk:2D-rect:diam2:VCdim} in \Cref{sec:2D-rect:diam3}.)

\paragraph{Graph preliminaries.}
Throughout this paper, let $G=(V,E)$ be a geometric intersection graph on $n$ geometric objects $\cO$. The input will be $\cO$ and the graph will be implicit.
We use $d(u,v)$ to denote the distance between two vertices $u,v\in V$ in $G$.
We denote the 1-hop \EMPH{neighborhood} of a vertex $v\in V$ by \EMPH{$N[v]$}, and the \EMPH{$r$-neighborhood ball} of $v$ by $\EMPH{$N^r[v]$} \coloneqq \{u\in V : d(u,v) \le r\}$.

\paragraph{Shatter Dimension.} For a set system $(U,\cF)$, define \EMPH{$Sh(n)$}$:=\max_{|S|=n}|\{S\cap F:F\in \cF\}|$. The \EMPH{shatter dimension} of the set system is the minimum $d$ such that $Sh(n) = O(n^d)$.
By the Sauer-Shelah Lemma~\cite{Sauer1972-ci,Shelah1972-ke}, if the VC-dimension (defined in the introduction) is $d$, then the shatter dimension is $d$.  
The converse partially holds: if the shatter dimension is $d$, then the VC-dimension is $O(d\log d)$.

\section{Fine Grained Lower Bounds}
\subsection{Fine grained complexity hypotheses}\label{appendix:hypothesis}

Recent development of fine-grained complexity has identified a few hypotheses that we include here for completeness. We use them to prove our lower bounds. 

\begin{definition}[Orthogonal Vectors (OV) hypothesis]\label{def:OV}
    Given sets $A, B$ of $n$ vectors in $\{0,1\}^d$, $d=\omega(\log n)$, deciding whether there exists an orthogonal pair $a\in A$, $b\in B$ requires $n^{2-o(1)}$ time.
\end{definition}
The OV-hypothesis is implied~\cite{Williams2005-nr} by the Strong Exponential Time Hypothesis~\cite{Impagliazzo2001-aw}.

\begin{definition}[3-uniform 6-hyperclique (3H6) hypothesis]\label{def:3H6}
Given a $6$-partite $3$-uniform hypergraph $G=(V,E)$ where $V$ is the disjoint union of vertex set $V^{(1)}, \ldots, V^{(6)}$, each containing $n$ vertices, and $E\subseteq {V \choose 3}$ such that each edge connects three vertices from different vertex sets. The problem is to decide whether there are $6$ vertices $S=\{v_1,v_2,\ldots,v_6\}$ with $v_i\in V^{(i)}$, $i=1,\ldots,6$, forming a $6$-clique, i.e., $\{v_i,v_j,v_k\}\in E$ for all $\{i,j,k\}\in {S\choose 3}$.
The 3H6 hypothesis says that the problem requires $n^{6-o(1)}$ time. 
\end{definition}

More information about the hyperclique hypothesis can be found in~\cite{Lincoln2018-ei}.

\begin{definition}[Combinatorial 4-clique (combK4) hypothesis]\label{def:combK4}
    Any combinatorial algorithm detecting whether a graph of $n$ vertices contains a $4$-clique requires $n^{4-o(1)}$ time.
\end{definition}

More information about the combinatorial $4$-clique problem and connections to other hypothesis can be found in~\cite{Chan2008-oe,Abboud2015-oc}.
\subsection{Diameter-3 lower bounds for unit balls in 3D}\label{sec:diameter3-balls}

\begin{theorem}\label{thm:3D-ball} Assuming the OV hypothesis,  there is no $O(n^{2-\eps})$ time algorithm for deciding if the intersection graph of a given set of $n$ unit balls in $\reals^3$ has diameter at most 3. 
\end{theorem} 

The reduction is from diameter-2 for $n$-vertex \EMPH{sparse tripartite graphs} $G = (A\cup B\cup C, E)$, where the number of edges $m$ is $\OO(n)$.
It is well-known that solving \Diameter-2 for sparse tripartite graphs in truly subquadratic time would break the OV hypothesis~\cite{Roditty2013-zr}. 
\begin{lemma}[\Diameter-2 sparse tripartite graphs]
\label{lm:diam-tripartite} 
There is no $O(n^{2-\eps})$ time algorithm for deciding whether a given sparse tripartite graph has diameter at most $2$, assuming the OV hypothesis.
\end{lemma}

In this section, we work with balls of radius $1/2$ instead of $1$.
Let $\eps = \Theta(1/n^4)$ and $\delta=\Theta(1/n^3)$.
First, we map every vertex $v$ in $A\cup C$ to a distinct number, also denoted by $v$,  in $[\eps,2\eps]$,
such that $\min_{a_1,a_2\in A:\, a_1\neq a_2} |a_1-a_2| \ge \Omega(\eps/n)$ and
$\min_{c_1,c_2\in C:\, c_1\neq c_2} |c_1-c_2| \ge \Omega(\eps/n)$. 
Vertices of $B$ are mapped to distinct numbers in $[0.9,1]$ such that 
$\min_{b_1,b_2\in B:\, b_1\neq b_2} |b_1-b_2| \ge \Omega(1/n)$.
Then we create four sets of balls as follows.
\begin{enumerate}
    \item  For every vertex $a\in A$, we add a ball centered at $s_a = (1+a\cos\delta,a\sin\delta,0)$.
    \item  For every edge $(a,b)\in (A\times B)\cap E$, we add a ball centered at $p_{ab} = (1+a\cos\delta-b\sin\delta,a\sin\delta + b\cos\delta,\sqrt{1-b^2})$.
    \item  For every edge $(b,c)\in (B\times C)\cap E$, we add a ball centered at $q_{bc} = (-c,b,\sqrt{1-b^2})$.
    \item  For every vertex $c\in C$, we add a ball centered at $t_c = (-c,0,0)$.
\end{enumerate}

\noindent Let \EMPH{$\mathcal{B}_i$} be the set of balls created in step $i$ for $i \in [4]$ and 
$\EMPH{$\mathcal{B}$} \coloneqq \cup_i \mathcal{B}_i$.   
Observe that:

\begin{observation}\label{obs:ball3D}Let \EMPH{$\ball{p}$} be the ball of radius $1/2$ centered at a point $p$. We have:
\begin{enumerate}
    \item $\ball{s_a}\cap \ball{p_{\Tilde{a}b}}\not= \emptyset$ if and only if $a = \Tilde{a}$.
    \item  $\ball{p_{ab}}\cap \ball{q_{\Tilde{b}c}}\not= \emptyset$ if and only if $b = \Tilde{b}$.
    \item  $\ball{p_{bc}}\cap \ball{t_{\Tilde{c}}}\not= \emptyset$ if and only if $c  = \Tilde{c}$.
    \item  Every two balls in $\mathcal{B}_i$ intersect for every $i \in [4]$.
    \item  For any two balls $b_1 \in \mathcal{B}_i$ and $b_2 \in \mathcal{B}_j$ such that $|i-j|\geq 2$, $b_1\cap b_2  = \emptyset$. 
\end{enumerate}
\end{observation}

\begin{proof}
For item 1, observe that
\begin{equation*}
\begin{split}
    \|s_a - p_{\Tilde{a}b}\|^2_2 &= ((a-\Tilde{a})\cos\delta + b\sin\delta)^2 + ((a-\Tilde{a})\sin\delta - b\cos\delta)^2 + 1 - b^2\\
    &= (a-\Tilde{a})^2(\cos^2\delta+\sin^2\delta) + b^2(\sin^2\delta+\cos^2\delta) + 1-b^2
    ~=~ (a-\Tilde{a})^2  + 1,
\end{split}
\end{equation*}
which is exactly 1 if $a=\Tilde{a}$, and at least $1+\Omega(1/n^{10})$ if $|a-\Tilde{a}|\ge \Omega(\eps/n)$.

\medskip
\noindent For item 2, observe that
\begin{equation*}
\begin{split}
    &\|p_{ab} - q_{\Tilde{b}c}\|^2_2 \\
    =& (1+a\cos\delta-b\sin\delta+c)^2 + 
    (a\sin\delta + b\cos\delta - \Tilde{b})^2 + (\sqrt{1 - b^2}-\sqrt{1-\Tilde{b}^2})^2\\
    =& (1+\Theta(\eps)-\Theta(\delta)+\Theta(\eps))^2 + (\Theta(\eps\delta) + b-\Theta(\delta^2)- \Tilde{b})^2 + (\sqrt{1 - b^2}-\sqrt{1-\Tilde{b}^2})^2\\
    =& (1-\Theta(1/n^3))^2 + (|b-\Tilde{b}| \pm \Theta(1/n^6))^2 + (\sqrt{1 - b^2}-\sqrt{1-\Tilde{b}^2})^2,
\end{split}
\end{equation*}
which is $1-\Theta(1/n^3)$ if $b=\Tilde{b}$, and at least $1+\Omega(1/n^2)$ if $|b-\Tilde{b}|\ge \Omega(1/n)$.

\medskip
\noindent For item 3, observe that
\(
    \|q_{bc}-t_{\Tilde{c}}\|^2_2 = (\Tilde{c}-c)^2 + 1,
\)
which is exactly 1 if $c=\Tilde{c}$, and at least $1+\Omega(1/n^{10})$ if $|c-\Tilde{c}|\ge \Omega(\eps/n)$.
\noindent Item 4 is easy to check.  
For item 5, observe that $\|s_a-t_c\|_2 \ge 1+a\cos\delta+c \ge 1+\Omega(1/n^4)$,
whereas $\|s_a - q_{bc}\|_2$ and $\|p_{ab}-t_c\|_2$ are both
$\sqrt{2}\pm o(1)$.
\end{proof}

In the real RAM model, balls $\mathcal{B}$ can clearly be constructed in $O(n+m) = \OO(n)$ time.
In the RAM model, it suffices to approximate all coordinates up to $O(1/n^{10})$ additive error, which can be done in $\OO(1)$ time (we could actually pick ``nice'' choices of $b\in B$ so that $b$ and $\sqrt{1-b^2}$ are both rational, and also a nice $\delta$ so that $\cos\delta$ and $\sin\delta$ are both rational).  Thus, the reduction can be done in $\OO(n)$ time. \Cref{thm:3D-ball} follows from the following lemma.

\begin{lemma}\label{lm:diameter3-ball} Let $K$ be the intersection graph of the set of balls $\mathcal{B}$. 
Then $G$ has diameter at most $2$ if and only if $K$ has diameter at most $3$.
\end{lemma}

\begin{proof}
Observe by item 4 of \Cref{obs:ball3D} that the balls in the same set $\mathcal{B}_i$ induce a clique in $K$. 
By item 5 of \Cref{obs:ball3D}, there is no edge between two balls in $\mathcal{B}_i$ and $\mathcal{B}_j$ when $|j-i| \geq 2$. Thus,
\emph{$K$ has diameter at most $3$ if and only if for every ball $\cube{s_a}\in \mathcal{B}_1$ and every ball $\cube{t_c} \in  \mathcal{B}_4$, $d_K(s_a, t_c)\leq 3$.}
Here we slightly abuse the notation by using the center of the ball to denote the corresponding vertex in $K$.  
 
Suppose that $G$ has diameter at most $2$.  Let $P = (a,b,c)$ be any path of length $2$ in $G$. 
Then by \Cref{obs:ball3D}, $(s_{a}, p_{ab}, q_{bc}, t_c)$ is a path of length 3 in $K$. Thus, $K$ has diameter at most 3 by the above observation. 
For the other direction, if $K$ has diameter at most $3$, then for any two vertices $s_a$ and $t_c$ (such that $\cube{s_a}\in \mathcal{B}_1$ and $\cube{t_c}\in \mathcal{B}_4$), a path of length 3 between $s_a$ and $t_c$ in $K$ must be of the form $(s_a, p_{\Tilde{a}b}, q_{\Tilde{b}c}, t_c)$. Then again by \Cref{obs:ball3D}, $a = \Tilde{a}$ and $b = \Tilde{b}$, and therefore, $(a,b,c)$ is a path of length $2$ in $G$.  
\end{proof}

\subsection{Diameter-3 lower bounds for unit hypercubes in 4D}

In this and the next two sections, we follow the same proof approach as in \Cref{sec:diameter3-balls} to establish conditional lower bounds for 
\Diameter-3 for other types of objects: 4D unit hypercubes, 3D cubes, and 2D rectangles.  The details are a little simpler than in the proof for 3D unit balls.

\begin{theorem}\label{thm:4D-cube} Assuming the OV hypothesis,  there is no $O(n^{2-\eps})$ time algorithm for deciding if the intersection graph of a given set of $n$ unit hypercubes in $\reals^4$ has diameter at most 3. 
\end{theorem} 

We provide a $\OO(n)$-time reduction from \Diameter-2 for sparse tripartite graphs. Let $G = (A\cup  B\cup C, E)$ be a sparse tripartite graph. For every point $c$, we denote by $\cube{c}$ the unit hypercube centered at $c$.  First, we map vertices of $A\cup B\cup C$ to (arbitrary) distinct numbers in $[0,0.1]$. We abuse the notation by using $v$ to denote the value that a vertex $v$ is mapped to. Then we create hypercubes in the following steps:
\begin{enumerate}
    \item For every vertex $a\in A$, we add a hypercube centered at $s_a = \{a,-a, 0,0\}$. 
    \item  For every edge $(a,b)\in (A\times B)\cap E$, we add a hypercube centered at $p_{ab} = \{1+a,1-a, 0.5+b, 0.5-b\}$. 
    \item  For every edge $(b,c)\in (B\times C)\cap E$, we add a hypercube centered at $q_{bc} = \{1+c, 1-c, 1.5 + b, 1.5-b\}$. 
    \item For every $c\in C$, we add a hypercube $t_c = \{c,-c,2,2\}$.   
\end{enumerate}

Let $\mathcal{Q}_i$ be the set of hypercubes created in step $i$ for $i \in [4]$.  Let $\mathcal{Q} = \mathcal{Q}_1\cup \mathcal{Q}_{2}\cup \mathcal{Q}_{3} \cup \mathcal{Q}_4$ be the set of resulting hypercubes.  The following observation explains our design of $\mathcal{Q}$.

\begin{observation}\label{obs:4Dcube} We have:
\begin{enumerate}
    \item  $\cube{s_{\Tilde{a}}}\cap \cube{p_{ab}}\not= \emptyset$ \emph{if and only if} $\Tilde{a} = a$.
    \item  $\cube{p_{a\Tilde{b}}}\cap \cube{q_{bc}}\not=\emptyset$ if and only if $b = \Tilde{b}$.
    \item $\cube{t_{\Tilde{c}}}\cap \cube{q_{bc}}\not=\emptyset$ if and only if $\Tilde{c} = c$. 
    \item Any two cubes in $\mathcal{Q}_i$ intersects for every $i \in [4]$.
     \item  Any two cubes $(x,y) \in \mathcal{Q}_i\times \mathcal{Q}_j$ for $i,j\in [4]$ such that $|i-j|\geq 2$ are disjoint: $x\cap y = \emptyset$. 
\end{enumerate}
\end{observation}

Clearly $\mathcal{Q}$ can be constructed in $O(n + m) = \OO(n)$ time. Thus, \Cref{thm:4D-cube} follows from the following lemma.

\begin{lemma}\label{lm:diameter3-cube} Let $K$ be the intersection graph of $\mathcal{Q}$. Then $G$ has diameter at most $2$ if and only if $K$ has diameter at most $3$.
\end{lemma}

\begin{proof} 
By the same argument in \Cref{lm:diameter3-ball} using \Cref{obs:4Dcube}.
\end{proof}

\subsection{Diameter-3 lower bound for 3D cubes}

\begin{theorem}\label{thm:3D-cube} Assuming the OV hypothesis,  there is no $O(n^{2-\eps})$ time algorithm for deciding if the intersection graph of a given set of $n$ cubes (all ``close'' to unit) in $\reals^3$ has diameter at most 3. 
\end{theorem} 

The reduction is again from \Diameter-2 for sparse tripartite graphs. Let $G = (A\cup B\cup C, E)$ be a tripartite graph with $n$ vertices and $m = \OO(n)$ edges. 

First, we map every vertex $v$ in $A\cup C$ to a distinct number, also denoted by $v$,  in $(0,\eps]$,
and every vertex $v$ in $B$ to a distinct number in $[1-\eps,1)$.
Then we create sets of cubes as follows.

\begin{enumerate}
    \item  For every vertex $a\in A$, we add a cube $s_a=[a,a+1]\times[a-1,a]\times[1,2]$.
    \item  For every edge $(a,b)\in (A\times B)\cap E$, we add a cube 
    $p_{ab} = [2a-b,a]\times [a,b]\times [b,2b-a]$ (with side length $b-a$).
    \item  For every edge $(b,c)\in (B\times C)\cap E$, we add a cube
    $q_{bc} = [2c-b,c]\times [b,2b-c]\times [c,b]$ (with side length $b-c$).
    \item  For every vertex $c\in C$, we add a cube $t_c=[c,c+1]\times [1,2]\times [c-1,c]$.
\end{enumerate}

The result follows as before.

\subsection{Diameter-3 lower bound for 2D rectangles}
\label{sec:2D-rect:diam3}
\begin{theorem}\label{thm:2D-rect:diam3} Assuming the OV hypothesis,  there is no $O(n^{2-\eps})$ time algorithm for deciding if the intersection graph of a given set of $n$ rectangles (all ``close'' to unit squares) in $\reals^2$ has diameter at most 3. 
\end{theorem} 

The reduction is again from \Diameter-2 for sparse tripartite graphs. Let $G = (A\cup B\cup C, E)$ be a tripartite graph with $n$ vertices and $m = \OO(n)$ edges. 

First, we map every vertex $v$ in $A\cup B\cup C$ to a distinct number, also denoted by $v$, in $(0,\eps]$.
Then we create sets of rectangles as follows.

\begin{enumerate}
    \item  For every vertex $a\in A$, we add a rectangle $s_a = [-1+a,a]\times [2+a,3+a]$.
    \item  For every edge $(a,b)\in (A\times B)\cap E$, we add a rectangle 
    $p_{ab} = [a,1+b]\times [1+b,2+a]$.
    \item  For every edge $(b,c)\in (B\times C)\cap E$, we add a rectangle
    $q_{bc} = [1+b,2+c]\times [c,1+b]$.
    \item  For every vertex $c\in C$, we add a rectangle $t_c=[2+c,3+c]\times [-1+c,c]$.
\end{enumerate}

\begin{remark} \label{rmk:2D-rect:diam2:VCdim}
A slight modification of this construction shows that diameter-2 rectangles have unbounded VC-dimension. Consider any set system $(\cS, X)$ and define a bipartite graph $G = (A\cup B, E)$ where $A=\cS$, $B = X$, and there is an edge from the vertex of $A$ corresponding to every set $S\in \cS$ to the vertex of $B$ corresponding to elements $x\in S$. Consider the rectangles $s_a$ and $p_{ab}$ from the first two steps of the above construction,
and additionally rectangles $q_b = [1+b,2+b] \times [b,1+b]$ for every $b\in B$. The diameter-2 balls centered at $s_a$ restricted to the squares $q_b$ corresponds exactly to $(\cS, X)$.
\end{remark}

\subsection{Diameter-2 lower bound for hypercubes}

We next turn to conditional lower bounds for \Diameter-2, using the hyperclique hypothesis instead of OV\@.  We revisit Bringmann \etal's result for 12D hypercubes~\cite{Bringmann2022-me}, and show how to slightly lower the number of dimensions to 10, by a more careful construction that uses some of the dimensions in a more economical way.

\begin{theorem}\label{thm:10D-cube} Assuming the 3-uniform 6-hyperclique hypothesis,  there is no $O(n^{2-\eps})$ time algorithm for deciding if the intersection graph of a given set of $n$ unit hypercubes in $\reals^{10}$ has diameter at most 2. 
\end{theorem} 

We provide a reduction from 3-uniform 6-hyperclique to \Diameter-2 for 10D unit hypercubes.
Let $G=(V_A\cup V_B\cup V_C\cup V_D\cup V_E\cup V_F,E)$ be a 6-partite 3-uniform hypergraph with $n^{1/3}$ vertices.  

First, we map every vertex $v$ in $V_A\cup V_B\cup V_C$ to a distinct number, also denoted by $v$, in $[0.3,0.4]$,
and every vertex $v$ in $V_D\cup V_E\cup V_F$ to a distinct number in $[0,0.1]$.
Then we create the following sets of $O(n)$ unit hypercubes in $\reals^{10}$ of side length 1 as follows.

\begin{enumerate}
    \item  For every hyperedge $(a,b,c)\in (V_A\times V_B\times V_C)\cap E$, we add a hypercube centered at 
    \[ s_{abc} = (a,1+a,\ b,1+b,\ c,1+c,\ 0.5,0.5,\ 0.5,0.5).\]  
    \item  For every hyperedge $(d,e,f)\in (V_D\times V_E\times V_F)\cap E$, we add a hypercube centered at 
    \[ t_{def} = (1+d,d,\ 1+d,d,\ 1+d,d,\ 1+e,e,\ 1+f,f).\]   
    \item  For every non-hyperedge $(a,b,d)\in (V_A\times V_B\times V_D)\setminus E$, we add a hypercube centered at
    \[ p_{abd} = (1+a,a,\ 1+b,b,\ d,1+d,\ 0.5,0.5,\ 0.5,0.5).\]
    \item  For every non-hyperedge $(a,b,e)\in (V_A\times V_B\times V_E)\setminus E$, we add a hypercube centered at
    \[ p_{abe} = (1+a,a,\ 1+b,b,\ 0.5,0.5,\ e,1+e,\ 0.5,0.5).\]
    \item  For every non-hyperedge $(a,b,f)\in (V_A\times V_B\times V_F)\setminus E$, we add a hypercube centered at
    \[ p_{abf} = (1+a,a,\ 1+b,b,\  0.5,0.5,\ 0.5,0.5,\ f,1+f).\]
    \item Similarly do the last 3 steps with $a,b$ replaced by $b,c$ or $a,c$.
    \item  For every non-hyperedge $(a,d,e)\in (V_A\times V_D\times V_E)\setminus E$, we add a hypercube centered at
    \[ p_{ade} = (1+a,a,\ 0.5,0.5,\ d,1+d,\ e,1+e,\ 0.5,0.5).\]
    \item  For every non-hyperedge $(a,d,f)\in (V_A\times V_D\times V_F)\setminus E$, we add a hypercube centered at
    \[ p_{adf} = (1+a,a,\ 0.5,0.5,\ d,1+d,\ 0.5,0.5,\ f,1+f).\]
    \item  For every non-hyperedge $(a,e,f)\in (V_A\times V_E\times V_F)\setminus E$, we add a hypercube centered at
    \[ p_{aef} = (1+a,a,\ 0.5,0.5,\ 0.5,0.5,\ e,1+e,\ f,1+f).\]
    \item Similarly do the last 3 steps with $a$ replaced by $b$ or $c$.
    \item We add an extra hypercube centered at $z_0=(0.8,0.2,\ 0.8,0.2,\ 0.8,0.2,\ 0.8,0.2,\ 0.8,0.2)$.
\end{enumerate}

It is not difficult to verify the following:

\begin{observation}\label{obs:10Dcube} We have:
\begin{enumerate}
    \item  $\cube{s_{abc}}\cap \cube{p_{\Tilde{a}\Tilde{b}d}}\not= \emptyset$ \emph{if and only if} $\Tilde{a} = a$ and $\Tilde{b}=b$.
    \item $\cube{t_{def}}\cap \cube{p_{ab\Tilde{d}}}\not= \emptyset$ if and only if $\Tilde{d}=d$.
    \item $\cube{p_{abd}}\cap\cube{z_0}\neq \emptyset$.
    \item  $\cube{s_{abc}}\cap\cube{t_{def}}=\emptyset$ and $\cube{s_{abc}}\cap \cube{z_0}=\emptyset$. 
\end{enumerate}
Similar statements hold for $p_{abe}$, etc.
\end{observation}

\begin{lemma}\label{lm:diameter2-10Dcube} $G$ contains a 6-hyperclique if and only if the intersection graph of the above boxes has diameter greater than $2$.
\end{lemma}
\begin{proof}
If $G$ contains a 6-hyperclique $(a,b,c,d,e,f)\in V_A\times V_B\times V_C\times V_D\times V_E\times V_F$, then
$\cube{s_{abc}}$ and $\cube{t_{def}}$ have distance more than 2 in the intersection graph,
since according to \Cref{obs:10Dcube} the only possible common neighbors of $\cube{s_{abc}}$ and $\cube{t_{def}}$ are
$\cube{p_{abd}}$, $\cube{p_{abe}}$, $\cube{p_{abf}}$, etc., but all these cubes are not present because $(a,b,d),(a,b,e),(a,b,f),\ldots\in E$.

Conversely, if two vertices have distance more than 2, then according to \Cref{obs:10Dcube} they must be of the form $\cube{s_{abc}}$ and $\cube{t_{def}}$ with $(a,b,c),(d,e,f)\in E$.  Since $\cube{p_{abd}}$, $\cube{p_{abe}}$, $\cube{p_{abf}}$, etc.\ are not common neighbors of $\cube{s_{abc}}$ and $\cube{t_{def}}$, we have   $(a,b,d),(a,b,e),(a,b,f),\ldots\in E$, and so $(a,b,c,d,e,f)$ is a 6-hyperclique in $G$.
\end{proof}

Thus, if \Diameter-2 for 10D unit hypercubes could be solved in truly subquadratic time, then we would have an algorithm for 6-hyperclique in a 3-uniform graph with $n^{1/3}$ vertices running in $O((n^{1/3})^{6-\eps})$ time, violating the 3H6 hypothesis.

\subsubsection{Variant 1}\label{sec:6D-cube}

We can modify the above proof to lower the dimension further to 6 for combinatorial algorithms, if we work under the combinatorial 4-clique hypothesis instead:

\begin{theorem}\label{thm:6D-cube} Assuming the combinatorial 4-clique hypothesis,  there is no $O(n^{2-\eps})$ time combinatorial algorithm for deciding if the intersection graph of a given set of $n$ unit hypercubes in $\reals^{6}$ has diameter at most 2. 
\end{theorem} 

We reduce from the 4-clique problem instead.  THe reduction can be viewed as a simplification of our previous reduction from 3-uniform 6-hyperclique.
Let $G=(A\cup B\cup C\cup D,E)$ be a 4-partite graph with $\sqrt{n}$ vertices.

First, we map every vertex $v$ in $A\cup B$ to a distinct number, also denoted by $v$, in $[0.3,0.4]$,
and every vertex $v$ in $C\cup D$ to a distinct number in $[0,0.1]$.
Then we create the following sets of $O(n)$ unit hypercubes in $\reals^{6}$ of side length 1 as follows.

\begin{enumerate}
    \item  For every edge $(a,b)\in (A\times B)\cap E$, we add a hypercube centered at 
    \[ s_{ab} = (a,1+a,\ b,1+b,\ 0.5,0.5).\]    
    \item  For every edge $(c,d)\in (C\times D)\cap E$, we add a hypercube centered at 
    \[ t_{cd} = (1+c,c,\ 1+c,c,\ 1+d,d).\]   
    \item  For every non-edge $(a,c)\in (A\times C)\setminus E$, we add a hypercube centered at
    \[ p_{ac} = (1+a,a,\ c,1+c,\ 0.5,0.5).\]
    \item  For every non-edge $(a,d)\in (A\times D)\setminus E$, we add a hypercube centered at
    \[ p_{ad} = (1+a,a,\ 0.5,0.5,\ d,1+d).\]
    \item  For every non-edge $(b,c)\in (B\times C)\setminus E$, we add a hypercube centered at
    \[ p_{bc} = (c,1+c,\ 1+b,b,\ 0.5,0.5).\]
    \item  For every non-edge $(b,d)\in (B\times D)\setminus E$, we add a hypercube centered at
    \[ p_{bd} = (0.5,0.5,\ 1+b,b,\ d,1+d).\]
   \item We add an extra hypercube centered at $z_0=(0.8,0.2,\ 0.8,0.2,\ 0.8,0.2)$.
\end{enumerate}

It is not difficult to verify the following:

\begin{observation}\label{obs:6Dcube} We have:
\begin{enumerate}
    \item  $\cube{s_{ab}}\cap \cube{p_{\Tilde{a}c}}\not= \emptyset$ \emph{if and only if} $\Tilde{a} = a$.
    \item $\cube{t_{cd}}\cap \cube{p_{c\Tilde{d}}}\not= \emptyset$ if and only if $\Tilde{d}=d$.
    \item $\cube{p_{ac}}\cap\cube{z_0}\neq \emptyset$.
    \item  $\cube{s_{ab}}\cap\cube{t_{cd}}=\emptyset$ and $\cube{s_{ab}}\cap \cube{z_0}=\emptyset$. 
\end{enumerate}
Similar statements hold for $p_{ad}$, $p_{bc}$, and $p_{bd}$.
\end{observation}

\begin{lemma}\label{lm:diameter2-6Dcube} $G$ contains a 4-clique if and only if the intersection graph of the above boxes has diameter greater than $2$.
\end{lemma}
\begin{proof}
    Similar to the proof of \Cref{lm:diameter2-10Dcube}, using \Cref{obs:6Dcube}.
\end{proof}

Thus, if \Diameter-2 for 6D unit hypercubes has a truly subquadratic time, combinatorial algorithm, then we would have a combinatorial algorithm for 6-hyperclique in a 3-uniform graph with $\sqrt{n}$ vertices  running in $O((\sqrt{n})^{4-\eps})$ time, violating the combinatorial 4-clique hypothesis.

\begin{remark}
For non-combinatorial algorithms, the above reduction still implies 
a non-trivial $\Omega(n^{3/2-\eps})$ lower bound, under the hypothesis that 4-clique for graphs with $n$ vertices requires near-cubic time.  (The fastest 4-clique algorithm known needs at least cubic time, even if the matrix multiplication exponent were to be 2.)
\end{remark}

\subsubsection{Variant 2}

Even more substantially, we can lower the dimension down to 4, but with a catch: the lower bound is weaker (though still superlinear for combinatorial algorithms).

\begin{theorem}\label{thm:4D-cube:diam2} Assuming the combinatorial 3-clique (or equivalently combinatorial BMM) hypothesis,  there is no $O(n^{3/2-\eps})$ time combinatorial algorithm for deciding if the intersection graph of a given set of $n$ unit hypercubes in $\reals^{4}$ has diameter at most 2.  The same holds for deciding if the maximum eccentricity of a given subset of size $\sqrt{n}$ in the intersection graph is at most 2.
\end{theorem} 

We reduce from the 3-clique problem, i.e., triangle detection.
The reduction can be viewed as a further simplification of the reduction from the previous subsection.
Let $G=(A\cup B\cup C,E)$ be a tripartite graph with $\sqrt{n}$ vertices.

First, we map every vertex $v$ in $A\cup B$ to a distinct number, also denoted by $v$, in $[0.3,0.4]$,
and every vertex $v$ in $C$ to a distinct number in $[0,0.1]$.
Then we create the following sets of $O(n)$ unit hypercubes in $\reals^{4}$ of side length 1 as follows.

\begin{enumerate}
    \item  For every edge $(a,b)\in (A\times B)\cap E$, we add a hypercube centered at 
    $s_{ab} = (a,1+a,\ b,1+b).$    
    \item  For every vertex $c\in C$, we add a hypercube centered at 
    $t_{c} = (1+c,c,\ 1+c,c).$ 
    \item  For every non-edge $(a,c)\in (A\times C)\setminus E$, we add a hypercube centered at
    $p_{ac} = (1+a,a,\ c,1+c).$
    \item  For every non-edge $(b,c)\in (B\times C)\setminus E$, we add a hypercube centered at
    $p_{bc} = (c,1+c,\ 1+b,b).$
   \item We add an extra hypercube centered at $z_0=(0.8,0.2,\ 0.8,0.2)$.
\end{enumerate}

As before, we can then show that $G$ contains a triangle if and only if the intersection graph of the above boxes has diameter greater than 2.
The second statement in the theorem also follows because the number of the $t_c$'s is only $\sqrt{n}$.  

\begin{remark}
Interestingly, although the first statement of \Cref{thm:4D-cube:diam2} may not be tight, the second is: we can compute the eccentricity of $\sqrt{n}$ vertices by running BFS $\sqrt{n}$ times, in total $\OO(n^{3/2})$ time (since one can implement BFS in the intersection graph of unit hypercubes in near linear time via orthogonal range searching data structures).
\end{remark}

\subsection{Diameter-2 Lower Bound for Unit Balls}

We can adapt our conditional lower bound proof for \Diameter-2 for 6D unit hypercubes in \Cref{sec:6D-cube} under the combinatorial 4-clique hypothesis, to the case of 7D unit balls:

\begin{theorem}\label{thm:7D-ball} Assuming the combinatorial 4-clique hypothesis,  there is no $O(n^{2-\eps})$ time combinatorial algorithm for deciding if the intersection graph of a given set of $n$ unit balls in $\reals^{7}$ has diameter at most 2. 
\end{theorem}

Let $G=(A\cup B\cup C\cup D,E)$ be a 4-partite graph with $\sqrt{n}$ vertices.

Let $\eps=\Theta(1/n^3)$.
First, we map every vertex $v$ in $A\cup B\cup C$ to a distinct number, also denoted by $v$, in $[0,0.1]$,
such that $\min_{a_1,a_2\in A:\, a_1\neq a_2} |a_1-a_2| \ge \Omega(1/n)$,
$\min_{b_1,b_2\in B:\, b_1\neq b_2} |b_1-b_2| \ge \Omega(1/n)$, and
$\min_{c_1,c_2\in C:\, c_1\neq c_2} |c_1-c_2| \ge \Omega(1/n)$.
Vertices of $D$ are mapped to distinct numbers in $[1-\eps,1]$ such that 
$\min_{d_1,d_2\in D:\, d_1\neq d_2} |d_1-d_2| \ge \Omega(\eps/n)$.
Then we create the following sets of $O(n)$ balls in $\reals^{7}$ of radius $\sqrt{2}$ as follows.

\begin{enumerate}
    \item  For every edge $(a,b)\in (A\times B)\cap E$, we add a ball centered at 
    \[ s_{ab} = (a,\sqrt{1-a^2},\ b,\sqrt{1-b^2},\ 0,0,\ 0).\]    
    \item  For every edge $(c,d)\in (C\times D)\cap E$, we add a ball centered at 
    \[ t_{cd} = (0,0,\ 0,0,\ c,\sqrt{1-c^2},\ d).\]   
    \item  For every non-edge $(a,c)\in (A\times C)\setminus E$, we add a ball centered at
    \[ p_{ac} = (a,\sqrt{1-a^2},\ 0,0,\ c,\sqrt{1-c^2},\ 0).\]
    \item  For every non-edge $(a,d)\in (A\times D)\setminus E$, we add a ball centered at
    \[ p_{ad} = (a,\sqrt{1-a^2},\ 0,0,\ 0,0,\ d).\]
    \item  For every non-edge $(b,c)\in (B\times C)\setminus E$, we add a ball centered at
    \[ p_{bc} = (0,0,\ b,\sqrt{1-b^2},\ c,\sqrt{1-c^2},\ 0).\]
    \item  For every non-edge $(b,d)\in (B\times D)\setminus E$, we add a ball centered at
    \[ p_{bd} = (0,0,\ b,\sqrt{1-b^2},\ 0,0,\ d).\]
   \item We add an extra ball centered at $z_1=(0,0,\ 0,0,\ 0.5,0.5,\ 0.5)$.
\end{enumerate}

It is not difficult to verify an analog of \Cref{obs:6Dcube}, and so the rest of the proof proceeds as before.

\section{A Near-Linear Diameter-2 Algorithm for 3D Unit Cubes}\label{sec:unitcube3d:diam2}

\newcommand{\Triv}{\top}
\newcommand{\SSS}{\mathcal{S}}

In this section, we present an $\OO(n)$-time algorithm for testing whether the diameter of a 3D unit cube graph is at most 2.  
This shows that the near-linear diameter-2 result of Bringmann \etal~\cite{Bringmann2022-me} for 2D unit squares surprisingly extends to 3D (ignoring logarithmic factors).  It also complements our lower bound result from \Cref{thm:4D-cube:diam2}, which shows conditionally that there are no similar, near-linear combinatorial diameter-2 algorithms for 4D unit hypercubes.

In \Cref{sec:unitcube3d:diam2:VC}, we study the combinatorial problem of bounding the VC-dimension of the distance-2 neighborhoods of 3D unit cube graphs as a warm-up.
In \Cref{sec:unitcube3d:diam2}, we use ideas from the combinatorial proof to design a near-linear-time divide-and-conquer~algorithm.

For a point $p\in\R^3$, denote \EMPH{$\cube{p}$} the unit cube centered at $p$.
For simplicity, we assume that the input is in general position, e.g., all coordinate values are distinct.
We will solve the problem in a slightly more general setting for 3 point sets $P,Q,R$, testing whether all distances between $\cube{p}$ and $\cube{r}$ for $(p,r)\in P\times R$ are exactly 2 in the tripartite intersection graph of the unit cubes centered at $P,Q,R$.  The original problem reduces to the case when $P,Q,R$ are all equal to the input point set (or almost equal, if we want to ensure general position); the diameter of the original point set is at most 2 iff every pair in $P\times R$ has distance exactly 2.

\subsection{VC-dimension bound}\label{sec:unitcube3d:diam2:VC}

For two points $p,q\in\R^3$, write \EMPH{$p\prec_x q$} if $p$ has smaller $x$-coordinate than $q$, and \EMPH{$p\succ_x q$} if $p$ has larger
$x$-coordinate than $q$.  Define \EMPH{$\prec_y$}, \EMPH{$\succ_y$}, \EMPH{$\prec_z$}, \EMPH{$\succ_z$} similarly.
Define the trivial relation \EMPH{$\Triv$} which is always true.
A \EMPH{generalized dominance relation} is a relation $\CMP$ where $p\CMP q$ iff $p\CMP_x q$ and $p\CMP_y q$ and $p\CMP_z q$, for
some choices of $\CMP_x\in\{\prec_x,\succ_x,\Triv\}$, $\CMP_y\in\{\prec_y,\succ_y,\Triv\}$, and $\CMP_z\in\{\prec_z,\succ_z,\Triv\}$.

The motivation for considering generalized dominance relations is this: Consider a uniform grid with unit side-length.  Inside a grid cell, any unit cube is equivalent to an orthant.  If $p$ lies in a grid cell and $q$ lies in a neighboring grid cell, then the condition that $\cube{p}$ and $\cube{q}$ intersect corresponds precisely to $p\CMP q$ for some generalized dominance relation.

Let $P$, $Q$, $R$ be 3 point sets in $\R^3$.  Let \EMPH{$\llhd=(\CMP_1,\CMP_2)$} be 2 generalized dominance relations in $\R^3$.
For each $r\in R$, we can write \EMPH{$N_{\llhd}^2[r]$}
as $\{p\in P: \exists q\in Q\ \mbox{with}\ p\CMP_1 q\ \textrm{and}\ q\CMP_2 r\}$.
Define the set system $\EMPH{$\SSS_{\llhd}(P,Q,R)$} \coloneqq (P,\,\{N^2_{\llhd}[r]: r\in R\})$.
We first prove that this set system has bounded VC-dimension by bounding its shatter dimension.

Previous proofs of bounded distance VC-dimension usually involve planarity arguments (avoidance of $K_5$), but for the 3D problem here, we use a different strategy.  We divide into easier cases based on the separability of the given point sets---in the following, we say that 
two sets are \EMPH{$x$-separated} if they are separated by a plane orthogonal to the $x$-axis; we define \EMPH{$y$-} and \EMPH{$z$-separation} similarly. We first show how to handle the case when $P$ and $Q$ are both $x$- and $y$-separated by a simple direct argument:

\begin{lemma}\label{lem:cube:diam2:VC1}
If $P$ and $Q$ are both $x$- and $y$-separated,
then $\SSS_{\llhd}(P,Q,R)$ has VC-dimension $1$.
\end{lemma}

\begin{proof}
Consider 2 points $p,p'\in P$.  Suppose $\{p\}$ and $\{p'\}$ are both shattered (by $\SSS_{\llhd}(P,Q,R)$).
Then there exist $r,r'\in R$ with $p\in N^2_{\llhd}[r]$, $p'\not\in N^2_{\llhd}[r]$, $p'\in N^2_{\llhd}[r']$, $p\not\in N^2_{\llhd}[r']$.
Let $q,q'\in Q$ with $p\CMP_1 q\CMP_2 r$ and $p'\CMP_1 q'\CMP_2 r'$.

For any $\CMP_{i}$, we define $\CMP_{ix}$, $\CMP_{iy}$, and $\CMP_{iz}$ to be the $x$-, $y$-, and $z$-coordinate dominance relation, respectively.
Because $P$ and $Q$ are $x$- and $y$-separated, we already know that
$p\CMP_{1x} q'$, $p\CMP_{1y} q'$, and $p'\CMP_{1x} q$, $p'\CMP_{1y} q$. 
Note that $p\CMP_{1z} q'$ or $p'\CMP_{1z} q$, because otherwise, we must have $\CMP_{1z}\neq \Triv$ and
$q'\CMP_{1z} p\CMP_{1z} q\CMP_{1z} p'\CMP_{1z} q'$: a contradiction.
Thus, $p\CMP_1 q'$ or $p'\CMP_1 q$, implying $p\in N^2_{\llhd}[r']$ or $p'\in N^2_{\llhd}[r]$.
\end{proof}

We then handle the case when $P$ and $R$ are separated along all 3 axes. 

\begin{lemma}\label{lem:cube:diam2:VC3}
If $P$ and $R$ are $x$-, $y$-, and $z$-separated,
then $\SSS_{\llhd}(P,Q,R)$ has shatter dimension at most $6$.
\end{lemma}

\begin{proof}
Without loss of generality, say $P\subset (-\infty,0)^3$ and $R\subset (0,\infty)^3$. 
Consider the following 6 set systems of the form $\SSS_{\llhd}(P, Q\cap\Lambda, R)$, where $\Lambda$ can be one of

\begin{tabular}{LLL}
((0,\infty)\times(0,\infty)\times\R), &
(\R\times(0,\infty)\times(0,\infty)), &
((0,\infty)\times\R\times(0,\infty)),
\\
((-\infty,0)\times(-\infty,0)\times\R), &
(\R\times(-\infty,0))\times(-\infty,0)), &
((-\infty,0)\times\R\times(-\infty,0)).
\end{tabular}
%
%

\noindent Each set system has VC-dimension at most 1 by Lemma~\ref{lem:cube:diam2:VC1} (or symmetric variants).

Each set in $\SSS_{\llhd}(P,Q,R)$ can be expressed as the union of 6 sets (one of each) from these 6 set systems
(since a point $q\in Q$ must has two positive coordinates or two negative coordinates).
So, the number of distinct sets in $\SSS_{\llhd}(P,Q,R)$ is at most $O(|P|^6)$.
\end{proof}

It remains to reduce the general case to the case when $P$ and $R$ are separated along all 3 axes.  We accomplish this by a simple grid idea:

\begin{lemma}\label{lem:cube:diam2:VC4}
For any $P,Q,R\subset\R^3$,
 $\SSS_{\llhd}(P,Q,R)$ has shatter dimension at most $51$.
\end{lemma}
\begin{proof}
Build a (non-uniform) grid formed by the $x$-, $y$-, and $z$-coordinates of all the points of $P$.
The grid has $O(|P|^3)$ cells.
For each grid cell interior $\gamma$, consider the 8 octants at an arbitrary point inside $\gamma$; for each such octant $\tau$, the set system $\SSS_{\llhd}(P\cap\tau,Q,R\cap\gamma)$ has shatter dimension at most 6 by \Cref{lem:cube:diam2:VC3}.
For each $r\in R\cap\gamma$, the set $N^2_{\llhd}[r]$ is the union of 8 sets (one of each) from the set systems for the 8 octants.
So for each $\gamma$, the number of distinct sets in $\{N^2_{\llhd}[r]: r\in R\cap\gamma\}$ is $O((|P|^6)^8)$, summing across all cells gives $|\{N^2_{\llhd}[r]: r\in R\}|  = O(|P|^3\cdot (|P|^{6})^8)$.
\end{proof}

We can now bound the shatter dimension for unit cubes.

\begin{lemma}
\label{lem:cube:diam2:VC5}
For any $P,Q,R\subset\R^3$ with $R\subset (0,1)^3$,
write $N^2[r]$ as
\[
\Set{\Big. p\in P: \exists q\in Q\ \mbox{with $\cube{q}$ intersecting both $\cube{p}$ and $\cube{r}$} }.
\]
Then the set system $(P, \{N^2[r]: r\in R\})$ has shatter dimension at most $4131$.
\end{lemma}

\begin{proof}
Fix $\alpha=(\alpha_P,\alpha_Q)\in(\Z^3)^2$ with $\|\alpha_P-\alpha_Q\|_\infty\le 1$.
For $p\in \alpha_P+(0,1)^3$ and $q\in \alpha_Q+(0,1)^3$, $\cube{p}$ intersects $\cube{q}$ iff $p-\alpha_P\CMP_1 q-\alpha_Q$,
where we define 
$\CMP_{1x}$ to be $\succ_x$ if $x(\alpha_Q)=x(\alpha_P)+1$,
and $\prec_x$ if $x(\alpha_Q)=x(\alpha_P)-1$, and $\Triv$ if 
$x(\alpha_Q)=x(\alpha_P)$;
we define $\CMP_{1y}$ and $\CMP_{1z}$ in the same way.
Similarly, for $q\in\alpha_Q+(0,1)^3$ and $r\in(0,1)^3$, $\cube{q}$ intersects $\cube{r}$ iff $q-\alpha_Q\CMP_2 r$
for some generalized dominance relation $\CMP_2$.
Define the set system $\SSS_{\alpha} \coloneqq \SSS_{\CMP_1,\CMP_2}(P\cap(\alpha_P+(0,1)^3) -\alpha_P ,(Q\cap(\alpha_Q + (0,1)^3)) -\alpha_Q, R)$.

Each set $N^2[r]$ is the union of at most $9^2$ sets (one of each) from
the set systems $\SSS_{\alpha}$ for the at most $9^2$ choices of $\alpha$.
By \Cref{lem:cube:diam2:VC4}, we have $|\{N^2[r]:r\in R\}|  = O((|P|^{51})^{9^2})$.
\end{proof}

\begin{theorem}
\label{thm:cube:diam2:VC}
For any $P,Q,R\subset\R^3$,
the set system $(P, \{N^2[r]: r\in R\})$ has shatter dimension at most $4131$.
\end{theorem}
\begin{proof}
Build a uniform grid of side length~1.
For each grid cell $\alpha_R+(0,1)^3$ (with $\alpha_R\in\Z^3$),
the number of different $N^2[r]$ sets over all $r\in R\cap (\alpha_R+(0,1)^3)$
is $O(|P\cap (\alpha_R+(-2,3)^3)|^{4131})$ by \Cref{lem:cube:diam2:VC5}.  
Since each point $p\in P$ belongs to $O(1)$ number of expanded cells $\alpha_R+(-2,3)^3$,
the sum over all $\alpha_R$ is $O(|P|^{4131})$.
\end{proof}

In the interest of simplicity, we have not optimized the above constant (which admittedly is quite large, but fortunately will not matter in our final algorithm).

\subsection{Algorithm}\label{sec:unitcube3d:diam2:main}

Knowing that the VC-dimension is bounded, we could at this point apply the framework of Chan et al.~\cite{ChanCGKLZ25} to obtain a subquadratic algorithm for \Diameter-2 of 3D unit cubes.  However, the exponent would be very close to 2 (much worse than our more general result for arbitrary 3D boxes in~\Cref{thm:box3d:diam2}).  We present a faster, direct algorithm that runs in \emph{near-linear} time.  The key observation is that although the final shatter dimension bound is large, the proof in \Cref{sec:unitcube3d:diam2:VC} tells us that the set system is in some sense ``made up of'' a constant number of simpler subsystems with much smaller VC-dimension, namely, VC-dimension~1 (\Cref{lem:cube:diam2:VC1})!  
We can't just solve the problem for each subsystem separately---the diameter problem isn't decomposable that way. Instead, we ``encode'' each subsystem by adding a single coordinate value to each element, and in the end reduce the whole problem to an orthogonal range searching problem for vectors in a sufficiently large constant dimension. 

As in \Cref{sec:unitcube3d:diam2:VC}, we begin with the corresponding problem for generalized dominance relations. As in \Cref{lem:cube:diam2:VC1}, we first consider the special case when $P$ and $Q$ are both $x$- and $y$-separated:

\begin{lemma}\label{lem:alg:cube1}
We can preprocess a point set $Q\subset\R^3$ in $\OO(|Q|)$ time so that the following holds.
Given point sets $P\subset (-\infty,\mu_x)\times (-\infty,\mu_y)\times \R$
and $R\subset \R^3$ for some $\mu_x,\mu_y$, and given generalized
dominance relations $\CMP_1$ and $\CMP_2$, we can compute mappings $\phi: P\rightarrow\R$
and $\psi: R\rightarrow\R$ in $\OO(|P|+|R|)$ time, satisfying the following property for every $(p,r)\in P\times R$: 
\begin{quote}
$(\exists q\in Q\cap ((\mu_x,\infty)\times (\mu_y,\infty)\times\R)$ with $p\CMP_1 q$ and $q\CMP_2 r)$\ \ $\Longleftrightarrow$\ \ $\phi(p)<\psi(r)$.
\end{quote}
\end{lemma}
\begin{proof}
We may assume that $\CMP_{1x}\in\{\prec_x,\Triv\}$ and $\CMP_{1y}\in\{\prec_y,\Triv\}$ (because if not,
we can trivially set $\phi=1$ and $\psi=0$).
We may assume that $\CMP_{1z}\neq\Triv$ (because if not,
we can replace all $z$-coordinates of $P$ with a sufficiently small negative number and replace $\CMP_{2z}$ with $\prec_z$).
If $\CMP_{1z}=\prec_z$, we define $\phi(p)$ to be the $z$-coordinate of $p$, and define $\psi(r)$ to be the largest $z$-coordinate among all points $q\in Q\cap ((\mu_x,\infty)\times (\mu_y,\infty)\times\R)$ with $q\CMP_2 r$.
The property is obviously satisfied.  Furthermore, $\psi$ can be evaluated in $\OO(1)$ time each by an orthogonal range max query~\cite{agarwal2017range},
assuming that $Q$ has been preprocessed in $\OO(|Q|)$ time (and $\phi$ is trivial to evaluate).  The case when   $\CMP_{1z}=\succ_z$ is similar (by negating all $z$-coordinates).
\end{proof}

Next, as in \Cref{lem:cube:diam2:VC3}, we consider the case when $P$ and $R$ are separated along all 3 axes:

\begin{lemma}\label{lem:alg:cube2}
We can preprocess a point set $Q\subset\R^3$ in $\OO(|Q|)$ time so that the following holds.
Given point sets $P\subset (-\infty,\mu_x)\times(-\infty,\mu_y)\times(-\infty,\mu_z)$
and $R\subset (\mu_x,\infty)\times(\mu_y,\infty)\times(\mu_z,\infty)$ for some $\mu_x,\mu_y,\mu_z$, and given generalized 
dominance relations $\CMP_1$ and $\CMP_2$, we can compute mappings $\phi: P\rightarrow\R^6$
and $\psi: R\rightarrow\R^6$ in $\OO(|P|+|R|)$ time, satisfying the following property for every $(p,r)\in P\times R$: 
\begin{quote}
$(\exists q\in Q$ with $p\CMP_1 q$ and $q\CMP_2 r)$\ \ $\Longleftrightarrow$\ \ $\phi(p)$ does not dominate $\psi(r)$.
\end{quote}
\end{lemma}
\begin{proof}
We compute mappings $\phi_1,\ldots,\phi_6: P\rightarrow\R$ and $\psi_1,\ldots,\psi_6: R\rightarrow\R$ satisfying the following properties:
\begin{enumerate}
\item ($\exists q\in Q\cap ((\mu_x,\infty)\times (\mu_y,\infty)\times\R)$ with $p\CMP_1 q$ and $q\CMP_2 r$) \ \ $\Longleftrightarrow$\ \ $\phi_1(p) < \psi_1(r)$;
\item ($\exists q\in Q\cap (\R\times (\mu_y,\infty)\times (\mu_z,\infty))$ with $p\CMP_1 q$ and $q\CMP_2 r$) \ \ $\Longleftrightarrow$\ \ $\phi_2(p) < \psi_2(r)$;
\item ($\exists q\in Q\cap ((\mu_x,\infty)\times \R\times (\mu_z,\infty))$ with $p\CMP_1 q$ and $q\CMP_2 r$) \ \ $\Longleftrightarrow$\ \ $\phi_3(p) < \psi_3(r)$;
\item ($\exists q\in Q\cap ((-\infty,\mu_x)\times (-\infty,\mu_y)\times\R)$ with $p\CMP_1 q$ and $q\CMP_2 r$)) \ \ $\Longleftrightarrow$\ \ $\phi_4(p) < \psi_4(r)$.
\item ($\exists q\in Q\cap (\R\times (-\infty,\mu_y)\times (-\infty,\mu_z))$ with $p\CMP_1 q$ and $q\CMP_2 r$) \ \ $\Longleftrightarrow$\ \ $\phi_5(p) < \psi_5(r)$;
\item ($\exists q\in Q\cap ((-\infty,\mu_x)\times\R\times (-\infty,\mu_z))$ with $p\CMP_1 q$ and $q\CMP_2 r$) \ \ $\Longleftrightarrow$\ \ $\phi_6(p) < \psi_6(r)$.
\end{enumerate}
Each such mapping can be computed by \Cref{lem:alg:cube1} (possibly with $x$-, $y$-, $z$-coordinates permuted and/or negated, and/or $P$ and $R$ swapped).  
Finally, we define $\phi(p)=(\phi_1(p),\ldots,\phi_6(p))$ and $\psi(r)=(\psi_1(r),\ldots,\psi_6(r))$.
\end{proof}

We now transform the result from dominance to unit cubes:

\begin{lemma}\label{lem:alg:cube3}
We can preprocess a point set $Q\subset\R^3$ in $\OO(|Q|)$ time so that the following holds.
Given point sets $P\subset \alpha_P + ((0,\mu_x)\times(0,\mu_y)\times(0,\mu_z))$
and $R\subset \alpha_R + ((\mu_x,1)\times (\mu_y,1)\times (\mu_z,1))$ for some $\mu_x,\mu_y,\mu_z\in (0,1)$ and
$\alpha_P,\alpha_R\in \Z^3$, we can compute mappings $\phi: P\rightarrow\R^{54}$
and $\psi: R\rightarrow\R^{54}$ in $\OO(|P|+|R|)$ time, satisfying the following property for every $(p,r)\in P\times R$: 
\begin{quote}
$(\exists q\in Q$ with $\cube{q}$ intersecting both $\cube{p}$ and $\cube{r})$ \ \ $\Longleftrightarrow$\ \ $\phi(p)$ does not dominate $\psi(r)$.
\end{quote}
\end{lemma}
\begin{proof}
For each $\alpha_Q\in \mathbb{Z}^3$ with $L_\infty$-distance at most 1 from both $\alpha_P$ and $\alpha_R$,
we compute a mapping $\phi^{(\alpha_Q)}$ satisfying the following property:
\begin{quote}
($\exists q\in Q\cap (\alpha_Q+(0,1)^3)$ with $\cube{q}$ intersecting both $\cube{p}$ and $\cube{r}$) \ \ $\Longleftrightarrow$\ \ $\phi^{(\alpha_Q)}(p) < \psi^{(\alpha_Q)}(r)$.
\end{quote}
Each such mapping can be computed by \Cref{lem:alg:cube2}.
This is because for $p\in \alpha_P+(0,1)^3$
and $q\in \alpha_Q+(0,1)^3$, $\cube{p}$ intersects $\cube{q}$ iff $p-\alpha_P\CMP_1 q-\alpha_Q$,
for some generalized dominance relation $\CMP_1$ as in the proof of \Cref{lem:cube:diam2:VC5}.
Similarly, for $q\in\alpha_Q+(0,1)^3$ and $r\in \alpha_R+(0,1)^3$, $\cube{q}$ intersects $\cube{r}$ iff $q-\alpha_Q\CMP_2 r-\alpha_R$ for some generalized
dominance relation $\CMP_2$.
Finally, we define $\phi$ and $\psi$ as the Cartesian products of $\phi^{(\alpha_Q)}$ and $\psi^{(\alpha_Q)}$ respectively over all
at most 9 choices of $\alpha_Q$.
\end{proof}

We then solve the problem by orthogonal range searching:

\begin{lemma}\label{lem:alg:cube4}
We can preprocess a point set $Q\subset\R^3$ in $\OO(|Q|)$ time so that the following holds.
Given point sets $P\subset \alpha_P + ((0,\mu_x)\times(0,\mu_y)\times(0,\mu_z))$
and $R\subset \alpha_R + ((\mu_x,1)\times (\mu_y,1)\times (\mu_z,1))$ for some $\mu_x,\mu_y,\mu_z\in (0,1)$ and
$\alpha_P,\alpha_R\in \Z^3$, 
we can decide whether for all $(p,r)\in P\times R$, 
there exists $q\in Q$ with $\cube{q}$ intersecting both $\cube{p}$ and $\cube{r}$,
in $\OO(|P|+|R|)$ time.
\end{lemma}
\begin{proof}
After constructing the mappings from \Cref{lem:alg:cube3}, we can test for the non-existence of
a pair $(p,r)\in P\times R$ with $\phi(p)$ dominating $\phi(r)$ and $p$ not intersecting $r$,
by $54$-dimensional orthogonal range searching in $\OO(|P|+|R|)$ time.
\end{proof}

Having solved the problem for the case when $P$ and $R$ (modulo 1) are separated along all 3 axes,
we still need to reduce the general problem to this case.
To this end, we could mimic the grid approach in the proof of \Cref{lem:cube:diam2:VC4}, but this would increase the running time (although this type of grid approach will still be useful in later sections, we more ambitiously aim for near-linear running time in this section).  Instead, we adopt a \emph{divide-and-conquer} algorithm to reduce to the separable case, which increases the running time only by a polylogarithmic factor.
The divide-and-conquer is similar in style to that of range trees~\cite{CGAA,AgarwalE99}, and may appear standard, but one unusual and interesting feature is that we are reducing only the sizes of $P$ and $R$ during recursion, while $Q$ stays the same.  This might seem problematic, but luckily in our setting, the cost at every recursive step is dependent only on $|P|$ and $|R|$ and not on $|Q|$, after an initial global preprocessing of $Q$ (we are not so lucky in our later algorithms).

\begin{lemma}\label{lem:alg:cube5}
We can preprocess a point set $Q\subset\R^3$ in $\OO(|Q|)$ time so that the following holds.
Given point sets $P\subset \alpha_P + (0,1)^3$ and $R\subset \alpha_R + (0,1)^3$ for some $\alpha_P,\alpha_R\in \Z^3$,
we can decide whether for all $(p,r)\in P\times R$, 
there exists $q\in Q$ with $\cube{q}$ intersecting both $\cube{p}$ and $\cube{r}$,
in $\OO(|P|+|R|)$ time.
\end{lemma}
\begin{proof}
We use ``range-tree-style'' divide-and-conquer.
First consider the special case when $P\subset \alpha_P + ((0,\mu_x)\times (0,\mu_y)\times (0,1))$
and $R\subset \alpha_R+ ((\mu_x,1)\times (\mu_y,1)\times (0,1))$, for a given $\mu_x,\mu_y$.
To solve the problem in this special case:
\begin{enumerate}
\item Let $\mu_z$ be the median $z$-coordinate in $(P-\alpha_P)\cup (R-\alpha_R)$.
Let $P^-= P\cap (\alpha_P + ((0,\mu_x)\times (0,\mu_y)\times (0,\mu_z)))$, $P^+=P\cap (\alpha_P+(0,\mu_x)\times (0,\mu_y)\times (\mu_z,1)))$,
$R^-=R\cap (\alpha_R+((0,\mu_x)\times (0,\mu_y)\times (0,\mu_z)))$, and $R^+=R\cap (\alpha_R+(0,\mu_x)\times (0,\mu_y)\times (\mu_z,1)))$.
\item Solve the problem for $P^-$ and $R^+$ by \Cref{lem:alg:cube4}.
\item Solve the problem for $P^+$ and $R^-$ by \Cref{lem:alg:cube4}
(after negating all $z$-coordinates).
\item Recursively solve the problem for $P^-$ and $R^-$.
\item Recursively solve the problem for $P^+$ and $R^+$.
\end{enumerate}
The running time for $m=|P|+|R|$ satisfies the recurrence $T_1(m)=2\,T_1(m/2)+\OO(m)$, which solves
to $T_1(m)=\OO(m)$.  (Note that $|Q|$ is not reduced during recursion, but luckily, 
the complexity does not depend on $|Q|$, excluding the initial preprocessing.)

Next consider the special case when $P\subset \alpha_P + ((0,\mu_x)\times (0,1)^2)$
and $R\subset \alpha_R + ((\mu_x,1)\times (0,1)^2)$, for a given $\mu_x$.
By a similar recursive algorithm via the median $y$-coordinate, we obtain
running time $T_2(m)=2\,T_2(m/2)+\OO(T_1(m))$, which solves to $T_2(m)=\OO(m)$.

Finally, the general case can be solved by another similar recursive algorithm via
the median $x$-coordinate, with running time $T(m)=2\,T(m/2)+\OO(T_2(m))$, which solves to $T(m)=\OO(m)$.
\end{proof}

\begin{theorem}\label{thm:3Dunitcube:diam2}
Given $3$ sets $P,Q,R$ of $O(n)$ points in $\R^3$,
we can decide whether for all $(p,r)\in P\times R$, 
there exists $q\in Q$ with $\cube{q}$ intersecting both $\cube{p}$ and $\cube{r}$,
in $\OO(n)$ time.
\end{theorem}
\begin{proof}
Build a uniform grid of side length 1.
For each nonempty grid cell $\alpha_P+(0,1)^3$ (with $\alpha_P\in\Z^3$) and for each
$\alpha_R\in \Z^3$ with $L_\infty$-distance at most 2 from $\alpha_P$,
we solve the problem for $P\cap (\alpha_P+(0,1)^3)$, $Q\cap (\alpha_P+(-1,2)^3)$, and 
$R\cap(\alpha_R+(0,1)^3)$
by \Cref{lem:alg:cube5} in time near-linear in the sizes of these three subsets.
Each point participates in only a constant number of subproblems.
\end{proof}

The number of logarithmic factors is admittedly huge (in the 50s), though we have not attempted to optimize it.

\section{Subquadratic Diameter-3 Algorithm for 3D Unit Cubes}\label{sec:diameter3-alg-unitcub}
\label{sec:unitcube3d:diam3}

\newcommand{\PPP}{\mathcal{P}}
\newcommand{\LLL}{\mathcal{L}}
\newcommand{\newCMP}{\blacktriangleleft}

In this section, we present a subquadratic algorithm for testing whether the diameter of a 3D unit cube graph is at most 3.
This complements our lower bound results,
which show conditionally that there are no similar diameter-3 algorithms for 3D unit ball graphs (\Cref{thm:3D-ball}) nor for 4D unit hypercube graphs (\Cref{thm:4D-cube}).
This result is more challenging than our result for diameter 2 in 
\Cref{sec:unitcube3d:diam2}.
As before, we warm up by studying the corresponding combinatorial problem of bounding the VC-dimension of the 3-neighborhoods of 3D unit cube graphs (\Cref{sec:unitcube3d:diam3:VC}).
The combinatorial proof will reveal a surprising connection to 2D pseudolines.
We review known techniques about pseudolines in \Cref{sec:pseudolines}, as preparation towards our final subquadratic algorithm, presented in \Cref{sec:unitcube3d:diam3:main}.

Again, for a point $p\in\R^3$, recall that \EMPH{$\cube{p}$} denotes the unit cube centered at $p$.
For simplicity, we assume that all coordinate values are distinct.
Like before, we will solve the problem in a slightly more general setting for 4 point sets $P,Q,R,S$, testing whether all distances between $\cube{p}$ and $\cube{s}$ for $(p,s)\in P\times S$ are at most~3 in the 4-partite intersection graphs of the unit cubes centered at $P,Q,R,S$.

\subsection{VC-dimension bound}\label{sec:unitcube3d:diam3:VC}

Following \Cref{sec:unitcube3d:diam2:VC}, we begin with the corresponding problem for generalized dominance relations.
Let $P,Q,R,S$ be four point sets in $\R^3$.  
Let \EMPH{$\lllhd ~= (\CMP_1,\CMP_2,\CMP_3)$}
be three generalized dominance relations in $\R^3$.
For each $s\in S$, we can write $N^3_{\lllhd}[s]$ as $\{p\in P: \exists (q,r)\in Q\times R\ \mbox{with}\ p\CMP_1 q\ \textrm{and}\ q\CMP_2 r\ \textrm{and}\ r\CMP_3 s\}$.
Define the set system $\EMPH{$\SSS_{\lllhd}$}(P,Q,R,S) \coloneqq (P,\,\{N^3_{\lllhd}[s]: s\in S\})$.
We first prove that this set system has bounded shatter dimension.

As in \Cref{sec:unitcube3d:diam2:VC}, we do not use planarity arguments like previous proofs. Instead, we divide into cases based on separability of the given point sets.  
But as we go from 2- to 3-neighborhoods, we face more challenges.  
The first case turns out to be the most crucial (and interesting), where as the one that follows it is similar to \Cref{lem:cube:diam2:VC3}.

\begin{lemma}\label{lem:cube:diam3:VC1}
If $P$ and $Q$ are $x$-separated, $Q$ and $R$ are $y$-separated, and $R$ and $S$ are $z$-separated, then
the set system $\SSS_{\lllhd}(P,Q,R,S)$ 
has VC-dimension at most $2$.
\end{lemma}
\begin{proof} 
Suppose $p\CMP_j q$ iff $p\CMP_{jx} q$ and $p\CMP_{jy} q$ and $p\CMP_{jz} q$ with $\CMP_{jx}\in\{\prec_x,\succ_x,\Triv\}$,
$\CMP_{jy}\in\{\prec_y,\succ_y,\Triv\}$, and $\CMP_{jz}\in\{\prec_z,\succ_z,\Triv\}$.  
We may assume that $\CMP_{1x}\neq\Triv$ (because if not, we can shift all the $x$-coordinates of $P$ sufficiently far downward and replace $\CMP_{1x}$ with $\prec_x$), and similarly, all other $\CMP_{jx},\CMP_{jy},\CMP_{jz}$ are not $\Triv$.
We first prove the following property:

\begin{quote}
\emph{Property}:~ Let $p,p'\in P$ and $s,s'\in S$.  If 
$p\in N^3_{\lllhd}[s]$, $p'\in N^3_{\lllhd}[s']$, $p\not\in N^3_{\lllhd}[s']$, $p'\not\in N^3_{\lllhd}[s]$, then $p\prec_y p'$
implies $s'\newCMP_y s$, where $\newCMP_y\in\{\prec_y,\succ_y\}$ is a relation determined solely from the choices of 
$\CMP_1,\CMP_2,\CMP_3$.
\end{quote}

\noindent\emph{Proof}:~
Let $q,q'\in Q$ and $r,r'\in R$ with $p\CMP_1 q\CMP_2 r\CMP_3 s$ and $p'\CMP_1 q'\CMP_2 r'\CMP_3 s'$.
We may assume neither $p\CMP_1 q'$ nor $p'\CMP_1 q$, because otherwise, $p\in N^3_{\lllhd}[s']$ or $p'\in N^3_{\lllhd}[s]$.
Because $P$ and $Q$ are $x$-separated, we already know $p\CMP_{1x} q'$ and $p'\CMP_{1x} q$.
So, there are two remaining possibilities: (i) $p\CMP_{1y} q\CMP_{1y} p'\CMP_{1y} q'$ and
$p'\CMP_{1z}q'\CMP_{1z} p\CMP_{1z}q$, or (ii) $p'\CMP_{1y} q'\CMP_{1y} p\CMP_{1y} q$ and
$p\CMP_{1z}q\CMP_{1z} p'\CMP_{1z}q'$.
(See Figure~\ref{fig:dom}.)

Define $\newCMP_{1z}\in\{\prec_z,\succ_z\}$ to be $\CMP_{1z}$ iff $\CMP_{1y}$ is $\prec_y$.  Then $p\prec_y p'$ implies
$q'\newCMP_{1z} q$.

\medskip
Similarly, we may assume neither $q\CMP_2 r'$ nor $q'\CMP_2 r$, because otherwise, $p\in N^3_{\lllhd}[s']$ or $p'\in N^3_{\lllhd}[s]$.
Because $Q$ and $R$ are $y$-separated, $q\CMP_{2y} r'$ and $q'\CMP_{2y} r$.
There are two remaining possibilities: (i) $q\CMP_{2z} r\CMP_{2z} q'\CMP_{2z} r'$ and
$q'\CMP_{2x}r'\CMP_{2x} q\CMP_{2x}r$, or (ii) $q'\CMP_{2z} r'\CMP_{2z} q\CMP_{2z} r$ and
$q\CMP_{2x}r\CMP_{2x} q'\CMP_{2x}r'$.

Define $\newCMP_{2x}\in\{\prec_x,\succ_x\}$ to be $\CMP_{2x}$ iff $\CMP_{2z}$ is $\newCMP_{1z}$.  Then $q'\newCMP_{1z} q$ implies
$r\newCMP_{2x} r'$.

\medskip
Similarly (and lastly), we may assume neither $r\CMP_3 s'$ nor $r'\CMP_3 s$, because otherwise, $p\in N^3_{\lllhd}[s']$ or $p'\in N^3_{\lllhd}[s]$.
Because $R$ and $S$ are $z$-separated, $r\CMP_{3z} s'$ and $r'\CMP_{3z} s$.
There are two remaining possibilities: (i) $r\CMP_{3x} s\CMP_{3x} r'\CMP_{3x} s'$ and
$r'\CMP_{3y}s'\CMP_{3y} r\CMP_{3y}s$, or (ii) $r'\CMP_{3x} s'\CMP_{3x} r\CMP_{3x} s$ and
$r\CMP_{3y}s\CMP_{3y} r'\CMP_{3y}s'$.

Define $\newCMP_y\in\{\prec_y,\succ_y\}$ to be $\CMP_{3y}$ iff $\CMP_{3x}$ is $\newCMP_{2x}$.  Then $r\newCMP_{2x} r'$ implies
$s'\newCMP_y s$. \hfill $\Box$

\begin{figure}
\centering
\includegraphics[scale=0.9]{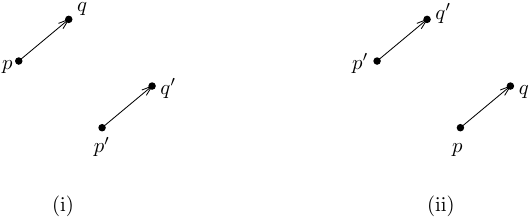}
\caption{
In the $yz$ plane, if $p$ is dominated by $q$ and $p'$ is dominated by $q'$, but $p$ is not dominated by $q'$ and
$p'$ is not dominated by $q$, then (i) and (ii) depict the only two possibilities.}
\label{fig:dom}
\end{figure}

\medskip
Now, consider 3 points $p,p',p''\in P$, with $p\prec_y p'\prec_y p''$.
Suppose $\{p,p''\}$ and $\{p'\}$ are both shattered.
Then there exist $s,s'\in S$ with $p,p''\in N^3_{\lllhd}[s]$, $p'\not\in N^3_{\lllhd}[s]$, $p'\in N^3_{\lllhd}[s']$, $p,p''\not\in N^3_{\lllhd}[s']$.
Applying the property twice yields $s'\newCMP_y s$ and $s\newCMP_y s'$: a contradiction.
\end{proof}


\begin{lemma}\label{lem:cube:diam3:VC2}
If (a) $P$, $Q$ are both $x$- and $y$-separated, or (b) $Q$, $R$ are both $x$- and $y$-separated, or
(c) $R$, $S$ are both $x$- and $y$-separated,
then $\SSS_{\lllhd}(P,Q,R,S)$ has VC-dimension~$1$.
\end{lemma}
\begin{proof}
Consider 2 points $p,p'\in P$.  Suppose $\{p\}$ and $\{p'\}$ are both shattered.
Then there exist $s,s'\in S$ with $p\in N^3_{\lllhd}[s]$, $p'\not\in N^3_{\lllhd}[s]$, $p'\in N^3_{\lllhd}[s']$, $p\not\in N^3_{\lllhd}[s']$.
Let $q,q'\in Q$ and $r,r'\in R$ with $p\CMP_1 q\CMP_2 r\CMP_3 s$ and $p'\CMP_1 q'\CMP_2 r'\CMP_3 s'$.

For (a), $p\CMP_{1x} q'$, $p\CMP_{1y} q'$, and $p'\CMP_{1x} q$, $p'\CMP_{1y} q'$ by the separation assumptions.
Note that $p\CMP_{1z} q'$ or $p'\CMP_{1z} q$, because otherwise,
$q'\CMP_{1z} p\CMP_{1z} q\CMP_{1z} p'\CMP_{1z} q'$: a contradiction.
Thus, $p\CMP_1 q'$ or $p'\CMP_1 q$, implying $p\in N^3_{\lllhd}[s']$ or $p'\in N^3_{\lllhd}[s]$.

Claims (b) and (c) follow from a similar argument. 
\end{proof}

We now handle the case when $P$ and $S$ are separated along all 3 axes, by invoking the previous cases some constant number of times:

\begin{lemma}\label{lem:cube:diam3:VC3}
If $P$ and $S$ are $x$-, $y$-, and $z$-separated,
then $\SSS_{\lllhd}(P,Q,R,S)$ has shatter dimension at most $21$.
\end{lemma}

\begin{proof}
Denote $\R_+$ and $\R_-$ to be $(0,\infty)$ and $(-\infty,0)$, respectively.
Without loss of generality, say $P\subset \R_-^3$ and $S\subset \R_+^3$.
Consider the following 15 set systems.

\begin{tabular}{LLLL}
\SSS_{\lllhd}(P, & Q\cap (\R_+\times\R_-\times\R), & R\cap (\R\times\R_+\times\R_-), & S)\\
\SSS_{\lllhd}(P, & Q\cap (\R\times\R_+\times\R_-), & R\cap (\R_-\times\R\times\R_+),  & S)\\
\SSS_{\lllhd}(P, & Q\cap (\R_-\times\R\times\R_+), & R\cap (\R_+\times\R_-\times\R),  & S)\\
\SSS_{\lllhd}(P, & Q\cap (\R_-\times\R_+\times\R), & R\cap (\R_+\times\R\times\R_-), & S)\\
\SSS_{\lllhd}(P, & Q\cap (\R\times\R_-\times\R_+), & R\cap (\R_-\times\R_+\times\R),  & S)\\
\SSS_{\lllhd}(P, & Q\cap (\R_+\times\R\times\R_-), & R\cap (\R\times\R_-\times\R_+),  & S) \\
\SSS_{\lllhd}(P, & Q\cap (\R_+\times\R_+\times\R), & R, & S)\\
\SSS_{\lllhd}(P, & Q\cap (\R_+\times\R\times\R_+), & R, & S)\\
\SSS_{\lllhd}(P, & Q\cap (\R\times\R_+\times\R_+), & R, & S)\\
\SSS_{\lllhd}(P, & Q\cap (\R_-\times\R_-\times\R), & R\cap (\R_+\times\R_+\times\R), & S)\\
\SSS_{\lllhd}(P, & Q\cap (\R_-\times\R\times\R_-), & R\cap (\R_+\times\R\times\R_+), & S)\\
\SSS_{\lllhd}(P, & Q\cap (\R\times\R_-\times\R_-), & R\cap (\R\times\R_+\times\R_+), & S)
\end{tabular}

\begin{tabular}{LLLL}
\SSS_{\lllhd}(P, & Q, & R\cap (\R_-\times\R_-\times\R), & S)\\
\SSS_{\lllhd}(P, & Q, & R\cap (\R_-\times\R\times\R_-), & S)\\
\SSS_{\lllhd}(P, & Q, & R\cap (\R\times\R_-\times\R_-), & S).
\end{tabular}
%

\noindent The first 6 set systems have VC-dimension at most 2 by Lemma~\ref{lem:cube:diam3:VC1} (or symmetric variants);
the next 9 have VC-dimension 1 by Lemma~\ref{lem:cube:diam3:VC2} (or symmetric variants).

Each set in $\SSS_{\lllhd}(P,Q,R,S)$ can be expressed as the union of 15 sets one from each of these 15 set systems: this is because, in any point sequence $\langle p,q,r,s\rangle$ with $p\in (-\infty,0)^3$ and $s\in (0,\infty)^3$, there must be a consecutive pair going from negative to positive $x$-coordinate, and a consecutive pair going from negative
to positive $y$-coordinate, and a consecutive pair going from negative
to positive $z$-coordinate; all 3 such pairs could be distinct, or 2 of them could be the same.
So, $|\SSS_{\lllhd}(P,Q,R,S)| = O((|P|^2)^6\cdot|P|^9)$.
\end{proof}

\noindent We can now reduce the general case to the case when $P$ and $S$ are separated along all 3~axes.

\begin{lemma}\label{lem:cube:diam3:VC4}
For any $P,Q,R,S\subset\R^3$,
 $\SSS_{\lllhd}(P,Q,R,S)$ has shatter dimension~$\le 171$.
\end{lemma}
\begin{proof}
Build a (non-uniform) grid formed by the $x$-, $y$-, and $z$-coordinates of all the points of $P$.
The grid has $O(|P|^3)$ cells.
For each grid cell $\gamma$, consider the 8 octants at an arbitrary point inside $\gamma$; for each such octant $\tau$, the set system $\SSS_{\lllhd}(P\cap\tau,Q,R,S\cap\gamma)$ has shatter dimension at most 18 by \Cref{lem:cube:diam3:VC3}.
For each $s\in S\cap\gamma$, the set $N^3_{\lllhd}[s]$ is the union of 8 sets (one of each) from the set systems for the 8 octants.
So for each $\gamma$, the number of distinct sets in $\{N^3_{\lllhd}[s]: s\in S\cap\gamma\}$ is $O((|P|^{21})^8)$. Thus, we have $|\{N^3_{\lllhd}[s]: s\in S\}| = O(|P|^3\cdot (|P|^{21})^8)$.
\end{proof}

As before, unit cubes reduce to generalized dominance relations:

\begin{lemma}
\label{lem:cube:diam3:VC5}
For any $P,Q,R,S\subset\R^3$ with $S\subset (0,1)^3$,
write $N^3[s]$ as
\[
\Set{\Big. p\in P: \exists (q,r)\in Q\times R\ \mbox{s.t.\ $\cube{p}$ intersects $\cube{q}$, $\cube{q}$ intersects $\cube{r}$, 
and $\cube{r}$ intersects $\cube{s}$} }.
\]
Then the set system $(P, \{N^3[s]: s\in S\})$ has shatter dimension at most $124659$.
\end{lemma}

\begin{proof}
Fix $\alpha=(\alpha_P,\alpha_Q,\alpha_R)\in(\Z^3)^3$ with $\|\alpha_R\|_\infty,\|\alpha_Q-\alpha_R\|_\infty,\|\alpha_P-\alpha_Q\|_\infty\le 1$.
For $p\in \alpha_P+(0,1)^3$ and $q\in \alpha_Q+(0,1)^3$, $\cube{p}$ intersects $\cube{q}$ iff $p-\alpha_P\CMP_1 q-\alpha_Q$,
for some generalized dominance relation $\CMP_1$ as in the proof of \Cref{lem:cube:diam2:VC5}.
Similarly, for $q\in\alpha_Q+(0,1)^3$ and $r\in\alpha_R+(0,1)^3$, $\cube{q}$ intersects $\cube{r}$ iff $q-\alpha_Q\CMP_2 r-\alpha_R$
for some generalized dominance relation $\CMP_2$.
Similarly, for $r\in\alpha_R+(0,1)^3$ and $s\in (0,1)^3$, $\cube{r}$ intersects $\cube{s}$ iff $r-\alpha_R\CMP_3 s$
for some generalized dominance relation $\CMP_3$.
Define the set system $\SSS_{\alpha} = \SSS_{\lllhd}(P\cap(\alpha_P+(0,1)^3)-\alpha_P, (Q\cap(\alpha_Q+ (0,1)^3))-\alpha_Q, (R\cap(\alpha_R+ (0,1)^3))-\alpha_R, S)$.

Each set $N^3[s]$ is the union of at most $9^3$ sets (one of each) from
the set systems $\SSS_{\alpha}$ for the at most $9^3$ choices of $\alpha$.
By \Cref{lem:cube:diam3:VC4},
 we have $|\{N^3[s]:s\in S\}| = O((|P|^{171})^{9^3})$.
\end{proof}

\begin{theorem}
\label{thm:cube:diam3:VC}
For any $P,Q,R,S\subset\R^3$,
the set system $(P, \{N^3[s]: s\in S\})$ has shatter dimension at most $124659$.
\end{theorem}

\begin{proof}
Build a uniform grid of side length~1.
For each grid cell $\alpha_S+(0,1)^3$ (with $\alpha_S\in\Z^3$),
the number of different $N^3[s]$ sets over all $s\in S\cap (\alpha_S+(0,1)^3)$
is $O(|P\cap (\alpha_S+(-3,4)^3)|^{124659})$ by \Cref{lem:cube:diam3:VC5}.
Since each point $p\in P$ belongs to $O(1)$ number of expanded cells $\alpha_S+(-3,4)^3$,
the sum over all $\alpha_S$ is $O(|P|^{124659})$.
\end{proof}

\subsection{Diameter-3 algorithms for 3D unit cubes}

Knowing that the VC-dimension is bounded, we could at this point apply the framework of~\cite{ChanCGKLZ25} to obtain a subquadratic algorithm for diameter-3 in 3D unit cubes.  However, the exponent would be \emph{extremely} close to 2 (around $2-1/(4\cdot 124660)\approx 1.999998$).  We present a faster, direct algorithm with a more reasonable exponent smaller than 2.  The key observation is that although the final shatter dimension bound is large, the proof in \Cref{sec:unitcube3d:diam3:VC} tells us that the set system is in some sense ``made up of'' a (very large but) constant number of simpler subsystems with much smaller VC-dimension, of~1 and 2 (\Cref{lem:cube:diam3:VC1} and \Cref{lem:cube:diam3:VC2}).  
As in our diameter-2 algorithm, we can't just solve the problem for each subsystem separately, because the diameter problem isn't decomposable.   Instead, we ``encode'' each subsystem as an extra ``2-dimensional constraint'', and in the end reduce the whole problem to a ``multi-level'' range searching problem in a sufficiently large constant dimension.  

Multi-level range searching~\cite{agarwal2017range,Matousek93} is usually solved by sampling-based geometric divide-and-conquer techniques, e.g., via so-called ``cuttings'', but
unfortunately analogs of cuttings provably do not exist for abstract set systems,
even with VC-dimension~2 (e.g., see 
\cite[Remark 4.7]{BhoreCHST25}).
Fortunately, our proof of VC-dimension~2 in \Cref{lem:cube:diam3:VC1} reveals a stronger property, namely, that
the 3-neighborhoods there actually form a \emph{point-pseudoline system}---see \Cref{sec:pseudolines}. (It is quite unexpected that 3D unit cubes could naturally produce pseudoline arrangements.)  As noted in \Cref{lem:cutting}, cuttings and multi-level range searching exists for pseudolines.

\subsubsection{Tools about pseudolines}\label{sec:pseudolines}

Before describing our subquadratic algorithm, we digress with some known combinatorial and computational facts about pseudoline arrangements.

\begin{definition}\rm
An \EMPH{abstract point-pseudoline system} consists of a pair of sets $(\PPP,\LLL)$.
Each element $p\in\PPP$
is a point in $\R^2$, and each element $\ell\in\LLL$ is a curve, where
the curves form a pseudoline family, i.e.,
each vertical line intersects each curve once, and each pair of curves
intersect at most once.
The curves are not explicitly given, but we assume that there are oracles to perform
the following operations:
\begin{quote}
\begin{enumerate}
\item[(O1)] Given $p\in\PPP$ and $\ell\in\LLL$, decide whether $p$ is below $\ell$.
\item[(O2)] Given $p\in\PPP$ and $\ell_1,\ell_2\in\LLL$, decide whether $\ell_1$ is below $\ell_2$ at the $x$-coordinate of $p$.
\end{enumerate}
\end{quote}
\end{definition}

\begin{lemma}\label{lem:pseudoline}
Consider a set system $(X,\SSS)$, where $X=\{p_1,\ldots,p_n\}$ and $\SSS=\{S_1,\ldots,S_m\}$,
satisfying the following condition:
\begin{quote}
\begin{enumerate}
\item[\rm ($\ast$)]
if $p_k\in S_j\setminus S_i$ and $p_h\in S_i\setminus S_j$ and $k<h$,
then $i<j$.
\end{enumerate}
\end{quote}
Then we can form an abstract point-pseudoline system $(\PPP,\LLL)$ and
mappings $\phi:X\rightarrow\PPP$ and $\psi:\SSS\rightarrow\LLL$, such that
$p_k\in S_i$ iff the point $\phi(p_k)$ is below the pseudoline $\psi(S_i)$.

Furthermore, operations (O1) and (O2) reduce to operations (O1$'$) and (O2$'$) respectively:
\begin{quote}
\begin{enumerate}
\item[\rm (O1$'$)] Given $i$ and $k$, decide whether $p_k\in S_i$.
\item[\rm (O2$'$)] Given $i<j$, find the smallest index $k$ with $p_k\in S_i\setminus S_j$.
\end{enumerate}
\end{quote}
\end{lemma}
\begin{proof}
The connection between set systems satisfying ($\ast$) and pseudolines is known before;
for example, see a work by Keszegh and Pálvölgyi~\cite{KeszeghP19}, who called such systems \emph{ABA-free hypergraphs}.  For completeness, we sketch a quick proof, so that one can see how operation (O2$'$) relates to (O2).
(A similar construction also appeared in a paper by Agarwal and Sharir~\cite{AgarwalS05}, in a different context about dualization of
pseudoline arrangements.)

We  define the points simply by $\phi(p_k)=(k,0)$.
We construct the curves $\psi(S_1),\ldots,\psi(S_m)$ from left to right, while ensuring that
$\psi(S_i)$ has positive $y$-coordinate at $x=k$ iff $p_k\in S_i$.
In $(-\infty,0]\times\R$, the curves $\psi(S_i)$ are non-intersecting and have negative $y$-coordinates in increasing order of $i$.
For $k=1,\ldots,n$, we do the following.
Consider four groups of curves: $A_k=\{\psi(S_i): p_{k-1},p_k\in S_i\}$, 
$B_k=\{\psi(S_i): p_{k-1}\in S_i,\ p_k\not\in S_i\}$,
$C_k=\{\psi(S_i): p_{k-1}\not\in S_i,\ p_k\in S_i\}$,
and  $D_k=\{\psi(S_i): p_{k-1},p_k\not\in S_i\}$.
Each group is non-intersecting in $[k-1,k]\times\R$.
At $x=k$, we make $D_k$ below $B_k$, below 0, below $C_k$, below $A_k$.
In $[k-1,k]\times\R$, a curve $\beta\in B_k$ intersects only the curves in $A_k$ that are below $\beta$ at $x=k-1$,
and intersects all the curves in $C_k$, but intersects none of the curves in $D_k$;
similarly, a curve $\gamma\in C_k$ intersects only the curves in $D_k$ that are above $\gamma$ at $x=k-1$, and
intersects all the curves in $B_k$, but intersects none of the curves in $A_k$.

The construction satisfies the following property: if $\psi(S_i)$ is below $\psi(S_j)$ at $x=k-1$, then
$\psi(S_i)$ intersects $\psi(S_j)$ in $[k-1,k]\times\R$ iff $p_k\in S_i\setminus S_j$.
It follows that if $i<j$, then $\psi(S_i)$ stays below $\psi(S_j)$ at all integer $x$-coordinates up to the smallest $k$
with $p_k\in S_i\setminus S_j$, then stays above at all larger $x$-coordinates, because of ($\ast$).
Hence, the constructed curves $\psi(S_i)$ form a pseudoline family.
Furthermore, for $i<j$, $\psi(S_i)$ is below $\psi(S_j)$ at $x$ (operation (O2)) iff $x$ is less than the smallest $k$ with $p_k\in S_i\setminus S_j$ (the index found by (O2$'$)). 
\end{proof}

\begin{lemma}\label{lem:cutting}
Let $(\PPP_1,\LLL_1),\ldots,(\PPP_D,\LLL_D)$ be abstract point-pseudoline
systems for a constant~$D$.  Given a set $P$ of size $m$, a set $S$ of size $n$, and mappings $\phi_j: P\rightarrow \PPP_j$ and $\psi_j:S\rightarrow\LLL_j$,
we can decide whether there exists $(p,s)\in P\times S$ such that $\phi_j(p)$ is above $\psi_j(s)$ for all $j\in\{1,\ldots,D\}$,
in $O^*(n\sqrt{C_1C_2m} + C_2m)$ expected time, where $C_1$ and $C_2$ are the costs of
operation (O1) and (O2) respectively, with $C_1\le C_2$.
Furthermore, we can report all such $(p,s)$ pairs in $O(K)$ additional time where $K$ is the output size.
\end{lemma}
\begin{proof}
This follows from standard ``multi-level'' range searching techniques via ``cuttings'' \cite{Chazelle93a,AgarwalS05} for pseudolines.  We re-sketch the approach below, and in particular, work out the dependence on $C_1$ and $C_2$ (which will be important in our application):

Take a random sample $R\subset S$ of size $c_0\rho\log\rho$,
and consider the vertical decomposition of the arrangement of the curves in $\psi_D(R)$,
which has $O(\rho^2\log^2\rho)$ size.  Each cell in the decomposition is a pseudo-trapezoid whose left and right sides are vertical
and the top and bottom sides are parts of the curves.  We shrink the pseudo-trapezoid inward so that the left
and right sides pass through points in $\phi_D(P)$.  This simplifies primitive operations.  (For example, two points in $\phi_D(P)$ are in the same pseudo-trapezoid iff they have the same set of curves below it, as well as the same curve immediately below it and the same curve immediately above it.
Thus, we can assign points in $\phi_D(P)$ to pseudo-trapezoids using $\OO(\rho m)$ number of operations (O1) and (O2).
Furthermore, we can test whether
a curve in $\psi_D(R)$ intersects a pseudo-trapezoid by making $O(1)$ number of operations (O1) and (O2).)

Standard Clarkson--Shor analysis \cite{ClarksonS89} tells us that with good probability, every pseudo-trapezoid intersects at most
$|S|/\rho$ curves, assuming that $c_0$ is a sufficiently large constant (we can re-sample when the condition is not met, for
an expected $O(1)$ number of trials).  By additional vertical cuts, we can ensure that every pseudo-trapezoid contains at most $|P|/\rho^2$ points, while increasing the number of pseudo-trapezoids by $O(\rho^2)$.  For each pseudo-trapezoid $\tau$, we recurse for the subset of all points inside $\tau$ and 
the subset of all curves intersecting $\tau$; we also recurse for the subset of all points inside $\tau$ and the subset of all curves completely below $\tau$, but with $D$ decremented.
This yields the following recurrence for the expected running time:
\[ T_D(m,n) \le O(\rho^2\log^2\rho) T_D(m/\rho^2,n/\rho) +  \OO(\rho^2 T_{D-1}(m,n) + \rho^2 C_2(m+n) ).\]
For the base case, when $m\le O(C_2/C_1)$, we switch to the naive bound $T_D(m,n)\le O(C_1 mn)\le O(C_2n)$.
Setting $\rho$ to be an arbitrarily large constant, 
the recurrence solves to $T_D(n,n) = O^*(n\sqrt{C_1C_2 m}+C_2m)$ by induction on $D$.

(Derandomization is possible via known techniques.  One could also get a better bound of the form $m^{2/3}n^{2/3}+m+n$, ignoring dependence on $C_1$ and $C_2$, if a dual pseudoline arrangement is available, but we will not need such an improvement.)
\end{proof}

\subsubsection{Algorithm}\label{sec:unitcube3d:diam3:main}

As in \Cref{sec:unitcube3d:diam3:VC}, we begin with the corresponding problem for generalized dominance relations. As in \Cref{lem:cube:diam3:VC1}, we first consider the most crucial special case when $P$ and $Q$ are $x$-separated, $Q$ and $R$ are $y$-separated, and $R$ and $S$ are $z$-separated:

\begin{lemma}\label{lem:unitcube3d:diam3:1}
Given point sets $P,Q,R,S$ in $\R^3$ of total size $n$, where
$P$ and $Q$ are $x$-separated, $Q$ and $R$ are $y$-separated, and $R$ and $S$ are $z$-separated,
and given 3 generalized dominance relations $\lllhd=(\CMP_1,\CMP_2,\CMP_3)$,
we can form an abstract point-pseudoline system $(\PPP,\LLL)$ 
mappings $\phi:P\rightarrow\PPP$ and $\psi:S\rightarrow\LLL$, satisfying the following property
for every $(p,s)\in P\times S$:
\begin{quote}
$(\exists (q,r)\in Q\times R:$ $p\CMP_1 q$ and $q\CMP_2 r$ and $r\CMP_3 s)$
\ \ $\Longleftrightarrow$\ \ $\phi(p)$ is below $\psi(s)$.
\end{quote}
After preprocessing $P,Q,R,S$ in $\OO(n|P|^{9/10})$ expected time, operation (O1) takes $\OO(1)$ time and operation (O2) takes $\OO(|P|^{4/5})$ time.
\end{lemma}
\begin{proof}
The proof of \Cref{lem:cube:diam3:VC1} actually shows the property ($\ast$) from \Cref{lem:pseudoline},
if we order $P$ by increasing $y$-coordinate and order $S$ by increasing/decreasing $y$-coordinate.
So, by Lemma~\ref{lem:pseudoline}, we obtain the point-pseudoline system $(\PPP,\LLL)$ and 
mappings $\phi,\psi$ with the desired property.

To implement the oracles, we apply the diameter computation technique by Chan et al.~\cite{ChanCGKLZ25}, knowing that the VC-dimension of the neighborhood sets is 2.
For each $q\in Q$, let $N^1_{\lllhd}[q] =\{ p\in P: p\CMP_1 q\}$.
For each $r\in R$, let $N^2_{\lllhd}[r] =\{ p\in P: p\CMP_1 q,\ q\CMP_2 r\}$.
For each $s\in S$, let $N^3_{\lllhd}[s] = \{p\in P: p\CMP_1 q,\ q\CMP_2 r,\ r\CMP_3 s\}$.
The algorithm computes the interval representations of the $N^1_{\lllhd}[q]$ sets, and then the interval representations of
the $N^2_{\lllhd}[r]$ sets from all the $N^1_{\lllhd}[q]$ sets, and finally the interval representations of all the $N^3_{\lllhd}[s]$ sets from the all $N^2_{\llhd}[r]$ sets.
The adaptation of Chan et al.'s technique here is straightforward:
\begin{itemize}
\item The algorithm needs to work with intermediate
set systems that have the VC-dimension doubled to 4, because our VC-dimension bound is for
the neighborhood sets of a fixed radius, not all radii simultaneously.  
\item We use a version of the algorithm via rainbow colored intersection searching: for 3D dominance,
it is straightforward to obtain such a data structure with near-linear preprocessing time and polylogarithmic query time,
since union of orthants has linear complexity (this is similar to the rainbow intersection searching data structure for 2D squares from \cite[Appendix~C.2]{ChanCGKLZ25}).  
\item Plugging into the analysis from~\cite{ChanCGKLZ25}, we get the following expected time bound when the diameter is constant:
\[ \OO(n\rho + n(|P|/\rho + \rho^4)b + n|P|/b),
\]
for parameters $b$ and $\rho$.  Setting $\rho=|P|^{1/5}$ and $b=|P|^{1/10}$ yields $\OO(n|P|^{9/10})$.
\end{itemize}

At the end, we have obtained interval representations
of all the $N^3_{\lllhd}[s]$ sets.  The total size of the interval representations is $\OO(|S|\cdot (|P|/\rho+\rho^4))=\OO(|S|\cdot |P|^{4/5})$.
(With more steps, we could lower it to $\OO(|S|\cdot \sqrt{|P|})$ since the actual VC-dimension of the 3-neighborhoods is 2, 
but this will not be important.)
Operation (O$1'$) is easy to support in $\OO(1)$ time: we just store the intervals of each $N^3_{\lllhd}[s]$ set in a binary search tree
during preprocessing.

Operation (O$2'$) requires more effort.
First, we store all the $y$-values of the elements along the stabbing path in a binary search tree for 1D range min/max queries.
For each $N^3_{\lllhd}[s]$ set,
we precompute the min/max $y$-value of the points in each interval of $N^3_{\lllhd}[s]$, 
and stores these values in another binary search tree; we do the same of the complement of $N^3_{\lllhd}[s]$.
Note that the interval representation of $N^3_{\lllhd}[s]$ has size $\OO(|P|^{4/5})$ for
at least half of the elements $s\in S$---call these elements \emph{good}.
We recursively solve the problem for the bad (i.e., non-good) elements of $S$ (keeping $P,Q,R$ the same).
This increases running time by a logarithmic factor.
(O$2'$) asks for the min/max $y$-value of elements in $N^3_{\lllhd}[s]\setminus N^3_{\lllhd}[s']$ for two given $s,s'\in S$.
If both $s$ and $s'$ are bad, we recurse.
Suppose one of $s$ and $s'$ is good---say it is $s$.  For each interval $I$ in $N^3_{\lllhd}[s]$, we find the min/max $y$-values for the intervals
in the complement of $N^3_{\lllhd}[s]$ that are completely inside $I$; this takes $\OO(1)$ time. There may still be two
remaining subintervals in $I\setminus N^3_{\lllhd}[s]$, and we can find the min/max $y$-values there in $\OO(1)$ time.
The overall query time is $\OO(1)$ times the number of intervals,
which is $\OO(|P|^{4/5})$ since $s$ is good.
The case when $s'$ is good is similar.
\end{proof}

Next, we consider the (easier) cases from \Cref{lem:cube:diam3:VC2}:

\begin{lemma}\label{lem:unitcube3d:diam3:2}
Given point sets $P,Q,R,S$ in $\R^3$ of total size $n$, where
$P$ and $Q$ are both $x$- and $y$-separated, or $Q$ and $R$ are both $x$- and $y$-separated, or
$R$ and $S$ are both $x$- and $y$-separated,
and given 3 generalized dominance relations $\lllhd=(\CMP_1,\CMP_2,\CMP_3)$,
we can compute mappings $\phi: P\rightarrow\R$
and $\psi: S\rightarrow\R$ in $\OO(n)$ time, satisfying the following property for every $(p,s)\in P\times S$: 
\begin{quote}
$(\exists (q,r)\in Q\times R:$ $p\CMP_1 q$ and $q\CMP_2 r$ and $r\CMP_3 s)$
\ \ $\Longleftrightarrow$\ \ $\phi(p)<\psi(s)$.
\end{quote}
\end{lemma}
\begin{proof}
This is similar to the proof of \Cref{lem:alg:cube1}.

Consider the case when $P$ and $Q$ are both $x$- and $y$-separated. W.l.o.g., say $P\in (-\infty,0)^2\times\R$ and
$Q\in (0,\infty)^2\times\R$.
We may assume that $\CMP_{1x}\in\{\prec_x,\Triv\}$ and $\CMP_{1y}\in\{\prec_y,\Triv\}$ (because if not,
we can trivially set $\phi=1$ and $\psi=0$).
We may assume that $\CMP_{1z}\neq\Triv$ (because if not,
we can replace all $z$-coordinates of $P$ with a sufficiently small negative number and replace $\CMP_{1z}$ with $\prec_z$).
Suppose $\CMP_{1z}=\prec_z$.  For each $p\in P$, we define $\phi(p)$ to be the $z$-coordinate of $s$.   Next, we define the weight of each point $r\in R$ to 
be the largest $z$-coordinate among all points $q\in Q$ with $q\CMP_1 r$;
these weights can be computed by $O(n)$ orthogonal range max queries;
finally, for each $s\in S$, we define $\psi(s)$ to be the largest weight among all points $r\in R$
with $r\CMP_1 s$; again these values can be computed by orthogonal range max queries.
The property is then satisfied.
The case when $\CMP_{1z}=\succ_z$ is similar (by negating all $z$-coordinates).

The case when $R$ and $S$ are both $x$- and $y$-separated is similar.

Finally consider the case when $Q$ and $R$ are both $x$- and $y$-separated.  W.l.o.g., say $Q\in (-\infty,0)^2\times\R$ and
$R\in (0,\infty)^2\times\R$.
As before, we may assume that $\CMP_{2x}\in\{\prec_x,\Triv\}$ and $\CMP_{2y}\in\{\prec_y,\Triv\}$.
We may assume that $\CMP_{2z}\neq\Triv$ (because if not,
we can shift all the $z$-coordinates of $R$ and $S$ upward by a large number and replace $\CMP_{2z}$ with $\prec_z$).
Suppose $\CMP_{2z}=\prec_z$.  For each $p\in P$, we define $\phi(p)$ to be the smallest $z$-coordinates among all points $q\in Q$
with $p\CMP_1 q$; these values can be computed by orthogonal range min queries.
For each $s\in S$, we define $\psi(s)$ to be the largest $z$-coordinates among all points $r\in R$ with $r\CMP_2 s$; again,
these values can be computed by orthogonal range max queries.  The property is then satisfied.
The case when $\CMP_{2z}=\succ_z$ is similar (by negating all $z$-coordinates).
\end{proof}

Following \Cref{lem:cube:diam3:VC3}, the natural next step would be to consider the case when $P$ and $S$ are separated along all 3 axes.  However, we will actually need to handle a slightly more general case.
For a box $\gamma=(\mu_x^-,\mu_x^+)\times (\mu_y^-,\mu_y^+)\times
(\mu_z^-,\mu_z^+)$, define
\[ \shadow(\gamma) := ((\mu_x^-,\mu_x^+)\times\R\times \R)\cup (\R\times(\mu_y^-,\mu_y^+)\times\R)\cup (\R\times\R\times (\mu_z^-,\mu_z^+)).\]
The case we will consider is when $P\subset \gamma$ and not all of $q,r,s$ are in $\shadow(\gamma)$:

\begin{lemma}\label{lem:unitcube3d:diam3:3}
Given point sets $P,Q,R,S$ in $\R^3$ of total size $n$, where $P\subset \gamma$ for some box $\gamma$,
and given 3 generalized dominance relations $\lllhd=(\CMP_1,\CMP_2,\CMP_3)$,
we can form $O(1)$ abstract point-pseudoline systems $(\PPP_j,\LLL_j)$ and 
mappings $\phi_j:P\rightarrow\PPP_j$ and $\psi_j:S\rightarrow\LLL_j$, satisfying the following property
for every $(p,s)\in P\times S$:
\begin{quote}
$(\exists (q,r)\in Q\times R:$ $p\CMP_1 q$ and $q\CMP_2 r$ and $r\CMP_3 s$, and
 not all of $q,r,s$ are in $\shadow(\gamma)$
\\[.5ex]$\Longleftrightarrow$\ \ $\phi_j(p)$ is below $\psi_j(s)$ for some $j$.
\end{quote}
After preprocessing $P,Q,R,S$ in $\OO(n|P|^{9/10})$ expected time, operation (O1) takes $\OO(1)$ time and operation (O2) takes 
$\OO(|P|^{4/5})$ time for each system.
\end{lemma}
\begin{proof}
The planes through the 6 faces of $\gamma$ divide $\R^3$ into
27 box cells.  Consider a triple $\tau=(\tau_Q,\tau_R,\tau_S)$ of cells such that not all of $\tau_Q,\tau_R,\tau_S$
are in $\shadow(\gamma)$; there are a (large) constant
number of such triples.
Then one of the pairs $(\tau_P,\tau_Q)$, $(\tau_Q,\tau_R)$, or $(\tau_Q,\tau_S)$ must be $x$-separated,
and  one of them must be $y$-separated,
and one of them must be $z$-separated.
We form a point-pseudoline system $(\PPP_\tau,\LLL_\tau)$ and mappings $\phi_\tau: P\rightarrow\PPP_\tau$ and $\psi_\tau: (S\cap\tau_S)\rightarrow\LLL_\tau$, satisfying the following property for every $(p,s)\in P\times (S\cap\tau_S)$:
\begin{quote}
$(\exists (q,r)\in (Q\cap\tau_Q)\times (R\cap\tau_R):$ $p\CMP_1 q$ and $q\CMP_2 r$ and $r\CMP_3 s)$
\ \ $\Longleftrightarrow$\ \ $\phi_\tau(p)$ is below $\psi_\tau(s)$.
\end{quote}
Such a system and mappings for each $\tau$ can be obtained from either \Cref{lem:unitcube3d:diam3:1} or \Cref{lem:unitcube3d:diam3:2} (in the latter case, numbers in $\R$ can be viewed as pseudolines in $\R^2$ trivially).
\end{proof}

We now transform the result from dominance to unit cubes:

\begin{lemma}\label{lem:unitcube3d:diam3:4}
Given point sets $P,Q,R,S$ in $\R^3$ of total size $n$, where $P\subset \gamma$ for some box $\gamma\subset (0,1)^3$,
we can form $O(1)$ abstract point-pseudoline systems $(\PPP_j,\LLL_j)$ and
mappings $\phi_j:P\rightarrow\PPP_j$ and $\psi_j:S\rightarrow\LLL_j$, satisfying the following property
for every $(p,s)\in P\times S$:
\begin{quote}
$(\exists (q,r)\in Q\times R:$ $\cube{p}$ intersects $\cube{q}$, $\cube{q}$ intersects $\cube{r}$, 
$\cube{r}$ intersects $\cube{s}$, and
 not all of $q,r,s$ are in $\shadow(\gamma)$ modulo $1$)
\\[.5ex]$\Longleftrightarrow$\ \ $\phi_j(p)$ is below $\psi_j(s)$ for some $j$.
\end{quote}
After preprocessing $P,Q,R,S$ in $\OO(n|P|^{9/10})$ expected time, operation (O1) takes $\OO(1)$ time and operation (O2) takes 
$\OO(|P|^{4/5})$ time for each system.
\end{lemma}
\begin{proof}
Consider a triple $\alpha=(\alpha_Q,\alpha_R,\alpha_S)\in (\Z^3)^3$ with $\|\alpha_Q\|_\infty,\|\alpha_P-\alpha_Q\|_\infty,
\|\alpha_R-\alpha_S\|_\infty\le 1$; there are a (large) constant number of such triples.
We form point-pseudoline systems $(\PPP^{(\alpha)}_j,\LLL^{(\alpha)}_j)$ and mappings $\phi^{(\alpha)}_j: P\rightarrow\PPP^{(\alpha)}_j$ and $\psi^{(\alpha)}_j: (S\cap(\alpha_S+(0,1)^3))\rightarrow\LLL^{(\alpha)}_j$, satisfying the following property for every $(p,s)\in P\times (S\cap (\alpha_S+(0,1)^3))$:
\begin{quote}
$(\exists (q,r)\in (Q\cap(\alpha_Q+(0,1)^3))\times (R\cap(\alpha_R+(0,1)^3)):$ $\cube{p}$ intersects $\cube{q}$, $\cube{q}$ intersects $\cube{r}$, 
$\cube{r}$ intersects $\cube{s}$, and not all $q,r,s$ are in $\shadow(\gamma)$ modulo $1$)\\[.5ex]
$\Longleftrightarrow$\ \ $\phi^{(\alpha)}_j(p)$ is below $\psi^{(\alpha)}_j(s)$ for some $j$.
\end{quote}
Such systems and mappings for each $\alpha$ can be obtained from \Cref{lem:unitcube3d:diam3:3} by defining
appropriate generalized dominance relations as in the proof of \Cref{lem:cube:diam2:VC5}.
\end{proof}

We then solve the problem by multi-level range searching:

\begin{lemma}\label{lem:unitcube3d:diam3:5}
Given point sets $P,Q,R,S$ in $\R^3$ of total size $n$, where $P\subset \gamma$ for some box $\gamma\subset (0,1)^3$,
we can decide whether
for all $(p,s)\in P\times S$,
there exists
$(q,r)\in Q\times R$ such that $\cube{p}$ intersects $\cube{q}$, $\cube{q}$ intersects $\cube{r}$, 
$\cube{r}$ intersects $\cube{s}$, and
 not all of $q,r,s$ are in $\shadow(\gamma)$ modulo $1$,
 in $O^*(n|P|^{9/10})$ expected time.
Furthermore, we can report all pairs $(p,s)$ not satisfying the property in $O(K)$ additional time, where $K$ is the output size.
 \end{lemma}
\begin{proof}
The problem can be solved by combining \Cref{lem:unitcube3d:diam3:4} with \Cref{lem:cutting}, where
$C_1=\OO(1)$ and $C_2=\OO(|P|^{4/5})$, in $O^*(n|P|^{9/10} + n\sqrt{|P|^{4/5}|P|})=O^*(n|P|^{9/10})$ expected time.
\end{proof}

We remark that the way we use range searching above is rather unusual (and interesting), as we are dealing with ranges/pseudolines that are not explicitly generated but  implicitly represented via oracles to (O1) and (O2).  (For one prior example, an algorithm by Chan on selection in totally monotone matrices~\cite{Chan21a} similarly dealt with ``abstract'' pseudolines, but there oracle costs are $\OO(1)$.  The closest example is perhaps Agarwal and Sharir's usage of abstract \emph{pseudo-disks} to solve the 2D discrete 2-center problem.)
In the recursive algorithm in the proof of \Cref{lem:cutting}, even though the number of objects corresponding to $P$ and $S$ may decrease, the oracle costs $C_1$ and $C_2$ stay fixed, as we cannot afford to re-preprocess (besides, $Q$ and $R$ never change).  Even if there might be room for improvement in the range searching part, the bottleneck for the $9/10$ exponent above lies in the simulation of the diameter algorithm of Chan \etal~\cite{ChanCGKLZ25} in \Cref{lem:unitcube3d:diam3:1}.

Finally, we reduce the general case to the case in \Cref{lem:unitcube3d:diam3:5}.  It is tempting to try a range-tree-style divide-and-conquer approach similar to the proof of \Cref{lem:alg:cube5}, but it does not seem to work, primarily because the size of $Q$ and $R$ do not necessarily decrease during recursion.  Instead, we switch to a grid approach (not exactly the same as that in \Cref{lem:cube:diam3:VC4});
this worsens the exponent, but not by much (from $2-1/10$ to $2-1/13$).

\begin{lemma}\label{lem:unitcube3d:diam3:6}
Given point sets $P,Q,R,S$ in $\R^3$ of total size $n$, where $P\subset (0,1)^3$,
we can decide whether
for all $(p,s)\in P\times S$,
there exists
$(q,r)\in Q\times R$ such that $\cube{p}$ intersects $\cube{q}$, $\cube{q}$ intersects $\cube{r}$, and 
$\cube{r}$ intersects $\cube{s}$, 
 in $O^*(n^{2-1/13})$ expected time.
\end{lemma}
\begin{proof}
Build a $g\times g\times g$ (nonuniform) grid over $(0,1)^3$, where every two consecutive parallel grid planes contain
$O(n/g)$ points in $P\cup Q\cup R\cup S$ modulo 1.
For each grid cell $\gamma$, let $P_\gamma=P\cap\gamma$ and do the following:

\newcommand{\ZZZ}{\mathcal{Z}}
\begin{enumerate}
\item Run the algorithm in \Cref{lem:unitcube3d:diam3:5} to report the list $L_\gamma$ of all pairs $(p,s)\in P_\gamma\times S$ violating  the following property:
there exists $(q,r)\in Q\times R$ such that $\cube{p}$ intersects $\cube{q}$, $\cube{q}$ intersects $\cube{r}$, 
$\cube{r}$ intersects $\cube{s}$, and not all of $q,r,s$ are in $\shadow(\gamma)$ modulo 1.  
This takes $O^*(n|P_\gamma|^{9/10}+ |L_\gamma|)$ time.
\item Next, find the list $L_\gamma'$ of all pairs $(p,s)\in P_\gamma\times S$ satisfying the following property: there exists
$(q,r)\in Q\times R$ such that $\cube{p}$ intersects $\cube{q}$, $\cube{q}$ intersects $\cube{r}$, 
$\cube{r}$ intersects $\cube{s}$, and
 $q,r,s$ are all in $\shadow(\gamma)$ modulo $1$.
 Since $\shadow(\gamma)$ contains only $O(n/g)$ points, we can solve the problem naively, by computing the 1-neighborhoods of each point $p\in P$,
 then the 2-neighborhoods, and finally the 3-neighborhoods, via orthogonal range searching,
 in $\OO(|P_\gamma|\cdot n/g)$ total time.
 \item Verify that $L_\gamma\subseteq L_\gamma'$, in $\OO(|L_\gamma'|)\le \OO(|P_\gamma|\cdot n/g)$ time.
\end{enumerate}
Note that since $|L_\gamma'|\le O(|P_\gamma|\cdot n/g)$, as soon as the number of elements in $L_\gamma$ we have found during step~1 exceeds a constant times $|P_\gamma|\cdot n/g$, we can stop and return false.

The total expected running time over all $O(g^3)$ grid cells $\gamma$ is
$O^*(\sum_\gamma (n|P_\gamma|^{9/10} + |P_\gamma|\cdot n/g)) = O^*(n^{19/10}(g^3)^{1/10} + n^2/g)$, which is $O^*(n^{2-1/13})$ by setting $g=n^{1/13}$.
\end{proof}

\begin{theorem}\label{thm:unitcube3d:diam3}
Given point sets $P,Q,R,S$ in $\R^3$ of total size $n$, 
we can decide whether
for all $(p,s)\in P\times S$,
there exists
$(q,r)\in Q\times R$ such that $\cube{p}$ intersects $\cube{q}$, $\cube{q}$ intersects $\cube{r}$, and 
$\cube{r}$ intersects $\cube{s}$, 
 in $O^*(n^{2-1/13})$ expected time.
\end{theorem}
\begin{proof}
Build a uniform grid of side length 1.
For each nonempty grid cell $\alpha_P+(0,1)^3$ (with $\alpha_P\in\Z^3$),
we solve the problem for $P\cap (\alpha_P+(0,1)^3)$, $Q\cap (\alpha_P+(-1,2)^3)$, and 
$R\cap(\alpha_P+(-2,3)^3)$, and $S\cap(\alpha_P+(-3,4)^3)$,
by \Cref{lem:unitcube3d:diam3:6}.
Each point participates in only a constant number of subproblems.
\end{proof}

\begin{remark}
We do not know how to generalize our subquadratic algorithm for 
3D unit cubes to diameter 4 and larger constants---this is perhaps one of the most intriguing questions we leave open.
However, for the related problem of designing efficient distance oracles, one could obtain
nontrivial results for distances up to 6: namely, there is a data structure with subquadratic preprocessing time and space, which can answer distance queries in sublinear (albeit, around $O(n^{0.999996})$) time.
This follows from our VC-dimension bound for 3-neighborhoods and the interval representation techniques from
\cite{ChanCGKLZ25}, since two vertices have distance at most 6 iff their 3-neighborhoods intersect.
\end{remark}

\section{Subquadratic Diameter-2 Algorithm for 3D Boxes}\label{sec:box3d:diam2:alg}

\newcommand{\ppp}{\hat{p}}
\newcommand{\pppp}{\hat{p}_{\downarrow}}
\newcommand{\UU}{\mathcal{U}}

In this section, we present a subquadratic algorithm for testing whether the diameter of an intersection graph for 3D boxes is at most 2.  This complements our lower bound result (\Cref{thm:3D-cube}), showing conditionally that there are no similar diameter-3 algorithms (this lower bound holds even in the special case of 3D cubes).  The running time of our algorithm is $\OO(n^{11/6})$.  Afterwards, we mention improvements in the special cases for 2D rectangles and 3D cubes (even in these special cases, subquadratic algorithms were not known before).

As before, we will solve the problem in a slightly more general setting for 3 sets $P,Q,R$ of objects, testing whether all distances between $P$ and $R$ are exactly 2 in the tripartite intersection graph.

In past sections, we started by bounding the VC-dimension of the corresponding set system, as warm-up to the design of a subquadratic algorithm.  However, for 3D boxes (or even 2D rectangles or 3D cubes), the VC-dimension is unbounded for distance-2 neighborhoods---thus, it is somewhat surprising that subquadratic algorithms are possible for objects as general as 3D boxes.

\subsection{Algorithm}

We use an approach based on a $g\times g\times g$ nonuniform grid for some choice of parameter $g$; this type of approach has
been used before to obtain subquadratic algorithms for other problems (e.g., $L_\infty$ discrete 3-center problem~\cite{ChanHY23}, 
or finding small-size independent set among rectangles or boxes~\cite{Chan23}).  For our problem, we face technical challenges due to the fact that 3D boxes may have \emph{quadratic union complexity}.
To resolve this issue, we show how to divide into a constant number of cases, so that in some cases, the boxes in the set $R$
can be replaced by orthants, and in the remaining cases, the boxes in the set $P$ can be replaced by orthants (exploiting the symmetry of the diameter problem).  Orthants in 3D are known to have linear union complexity.

For a box $q\subset\R^3$, let $x^-(q)$, $y^-(q)$, and $z^-(q)$ denote its min $x$-, $y$-, and $z$-coordinate respectively,
and  $x^+(q)$, $y^+(q)$, and $z^+(q)$ denote its max $x$-, $y$-, and $z$-coordinate respectively.
For simplicity, we assume that all coordinate values are distinct.

Map each box $q\subset\R^3$ to a point $\phi(q)=(x^-(q),-x^+(q),y^-(q),-y^+(q),z^-(q),-z^+(q))\in\R^6$.
Map each box $p\subset\R^3$ to another point $\psi(p)=(x^+(p),-x^-(p),y^+(p),-y^-(p),z^+(p),-z^-(p))\in\R^6$.
Observe that $p$ and $q$ intersect iff $\phi(q)\prec \psi(p)$, where $\prec$ denotes dominance.

For each $I\subset\{1,\ldots,6\}$, let $\pi_I:\R^6\rightarrow\R^{|I|}$ denote the projection map where we keep only the
coordinate positions in $I$.

\begin{lemma}\label{lem:box3d:diam2}
Let $I\subset\{1,\ldots,6\}$ be of size $3$.
Given $3$ sets $P,Q,R$ of $O(n)$ boxes in $\R^3$,
we can decide whether for all $(p,r)\in P\times R$ with $\pi_I(\psi(p))\prec \pi_I(\psi(r))$,
there exists $q\in Q$ intersecting both $p$ and $r$, in $\OO(n^{11/6})$ time.
\end{lemma}
\begin{proof}
Form an $g\times g\times g$ (nonuniform) grid over $\R^3$, where there are $O(n/g)$ box vertices
between any two consecutive parallel grid planes, for a parameter $g$ to be set later.  
A \emph{grid box} refers to a box whose sides lie on grid planes; there are $O(g^6)$ possible grid boxes.
For each box $p$, let $\ppp$ be the largest grid box contained in $p$.

For each grid box $\ppp$,
precompute a set
\[ R(\ppp) = \{r\in R: \mbox{$\exists q\in Q$ intersecting both $\ppp$ and $r$}\}.
\]
We can do so by first computing $Q(\ppp)=\{q\in Q: \mbox{$q$ intersects $\ppp$}\}$ naively in $O(n)$ time per $\ppp$,
and then computing $R(\ppp)=\{r\in R: \mbox{$r$ intersects some $q\in Q(\ppp)$}\}$ by performing $O(n)$ orthogonal
range intersection queries on $Q(\ppp)$, in $\OO(n)$ total time per $\ppp$.  The total time so far is $\OO(g^6n)$.

Fix $p\in P$.  We want to check the following condition:
\begin{center}
$\forall r\in R\setminus R(\ppp)$ with $\pi_I(\psi(p))\prec \pi_I(\psi(r))$:\ \ $\exists q\in Q$ intersecting both $p$ and~$r$.
\end{center}
Let $L(p) := \{q\in Q: \mbox{$q$ intersects $p$ but does not intersect $\ppp$}\} = \{q\in Q: \mbox{$\partial q$ intersects $p\setminus\ppp$}\}$, which has cardinality at most $O(n/g)$.
The condition is then equivalent to:
\begin{center}
$\forall r\in R\setminus R(\ppp)$ with $\pi_I(\psi(p))\prec \pi_I(\psi(r))$:\ \ $\exists q\in L(p)$ intersecting~$r$.
\end{center}
If we already know $q$ intersects $p$, then
$
\mbox{$q$ intersects $r$}\ \Longleftrightarrow\ \phi(q)\prec \psi(r)\ \Longleftrightarrow\  \pi_{I^c}(\phi(q))\prec \pi_{I^c}(\psi(r)),
$
where $I^c=\{1,\ldots,6\}\setminus I$\ (which has size~3).
This is because $\pi_I(\phi(q))\prec \pi_I(\psi(p))\prec \pi_I(\psi(r))$.
Thus, the above condition can be rewritten as:
\begin{center}
$\forall r\in R\setminus R(\ppp)$ with $\pi_I(\psi(p))\prec \pi_I(\psi(r))$:\ \ 
$\pi_{I^c}(\psi(r))\in \UU(p)$, where $\UU(p) := \displaystyle\bigcup_{q\in L(p)} \{\xi\in\R^3: \pi_{I^c}(\psi(q))\prec \xi\}$.
\end{center}
Now, $\UU(p)$ is a union of $O(n/g)$ orthants in $\R^3$, forming a ``staircase'' polyhedron of $O(n/g)$ complexity.
We can compute $\UU(p)$ and decompose the complement of $\UU(p)$ into $O(n/g)$ box cells, in $\OO(n/g)$ time~\cite{AgarwalSS24}.
To check the above condition, it suffices to examine each such cell $\gamma$ and test
whether there exists $r\in R\setminus R(\ppp)$ with $\pi_I(\psi(p))\prec \pi_I(\psi(r))$ and $\pi_{I^c}(\psi(r))\in\gamma$.
This can be done by performing $O(n/g)$ constant-dimensional orthogonal range queries, in 
$\OO(n/g)$ time, after preprocessing $R-R(\ppp)$ in $\OO(n)$ time.
The total preprocessing time for the range queries is $\OO(g^6n)$, and the total query time over all $s\in S$
is $\OO(n\cdot n/g)$.  Setting $g=n^{1/7}$ gives an $\OO(n^{13/7})$ time bound.

To improve the bound, we use \emph{bottomless grid boxes}, i.e., grid boxes that are unbounded from below in the $z$-direction.
There are only $O(g^5)$ bottomless grid boxes.
For each box $p$, define $\pppp$ to be the bottomless grid box formed by extending $\ppp$ downward.
For each bottomless grid box $\pppp$, compute a weighted
point set $Q(\pppp)=\{q\in Q: \mbox{$q$ intersects $\pppp$}\}$ naively in $O(n)$ time, where
the weight of $q$ is the largest $z$ such that $q$ intersects $\pppp\cap (z,\infty)$.
Compute a weighted point set $R(\pppp)=\{r\in R: \mbox{$r$ intersects some $q\in Q(\pppp)$}\}$, where
the weight of $r$ is the largest weight of all $q\in Q(\pppp)$ intersecting $r$;
this can be done by performing orthogonal range max queries, in $\OO(n)$ time per $\pppp$.
The time for this step is $\OO(g^5n)$.
These $O(g^5)$ weighted point sets provide an implicit representation of $R(\ppp)$ for all $O(g^6)$ grid boxes $\ppp$, 
since $R(\ppp)$ is just the subset of all points of $R(\pppp)$ with weight greater than $z^-(\ppp)$.
After preprocessing $R\setminus R(\pppp)$, we can proceed as before, testing each $p\in P$ using $\OO(n/g)$ range queries---these are now range min/max queries.
The overall time bound is now $\OO(g^5n + n\cdot n/g)$, which is $\OO(n^{11/6})$ by setting $g=n^{1/6}$.
\end{proof}

\begin{theorem}\label{thm:box3d:diam2}
Given $3$ sets $P,Q,R$ of $O(n)$ boxes in $\R^3$,
we can decide whether for all $(p,r)\in P\times R$, 
there exists $q\in Q$ intersecting both $p$ and $r$, in $\OO(n^{11/6})$ time.
\end{theorem}
\begin{proof}
For any two boxes $p$ and $r$, we must have (i)~$\pi_I(\psi(p))\prec \pi_I(\psi(r))$ for some $I\subset\{1,\ldots,6\}$ of size 3,
or (ii)~$\pi_I(\psi(r))\prec \pi_I(\psi(p))$ for some $I\subset\{1,\ldots,6\}$ of size 3.  (This is because if $\psi(p)$ is less than $\psi(r)$ in fewer than 3 coordinate positions,
then $\psi(r)$ is less than $\psi(p)$ in more than 3 coordinate positions.)
Thus, it suffices to run the algorithm in \Cref{lem:box3d:diam2} for each $I$ of size 3, for $P$ and $R$, and also for $P$ and $R$ swapped.
\end{proof}

\subsection{Special cases}

\begin{theorem}\label{thm:rect:diam2}
Given $3$ sets $P,Q,R$ of $O(n)$ rectangles in $\R^2$,
we can decide whether for all $(p,r)\in P\times R$, 
there exists $q\in Q$ intersecting both $p$ and $r$, in $\OO(n^{7/4})$ time.
\end{theorem}
\begin{proof}
The algorithm is similar (and simpler), except that with a 2-dimensional grid, the number of bottomless grid rectangles is $O(g^3)$.
The overall time bound is now $\OO(g^3n + n\cdot n/g)$, which is $\OO(n^{7/4})$ by setting $g=n^{1/4}$.
\end{proof}

\begin{theorem}\label{thm:cube3d:diam2}
Given $3$ sets $P,Q,R$ of $O(n)$ cubes in $\R^3$,
we can decide whether for all $(p,r)\in P\times R$, 
there exists $q\in Q$ intersecting both $p$ and $r$, in $\OO(n^{9/5})$ time.
\end{theorem}
\begin{proof}
The algorithm is the same (but without the improvement via bottomless grid boxes $\pppp$),  We observe that for boxes $p$ that are cubes, the number of possible choices of grid boxes $\ppp$ is actually $O(g^4)$ instead of $O(g^6)$.
To see this, map the boxes $p$ to points $\phi(p)$ in $\R^6$.
Boxes $p$ with a common $\ppp$ map to points in a common grid cell in $\R^6$, for a (non-uniform) grid defined by
$O(g)$ (axis-parallel) grid hyperplanes.
Boxes $p$ that are cubes map to points lying on a 4-dimensional flat $h=\{(x,-(x+w),y,-(y+w),z,-(z+w))\in\R^6: x,y,z,w\in\R\}$.
The $O(g)$ grid hyperplanes form a 4-dimensional arrangement of complexity $O(g^4)$ inside $h$.
Thus, the number of grid cells intersecting $h$ is only $O(g^4)$.
The overall time bound is now $\OO(g^4n + n\cdot n/g)$, which is $\OO(n^{9/5})$ by setting $g=n^{1/5}$.
\end{proof}

\small
\bibliographystyle{alphaurl}
\bibliography{refs}

\end{document}